\documentclass[a4paper,11pt]{article}
\let\counterwithout\relax

\usepackage{graphicx,subcaption}
\usepackage{mathrsfs,amssymb,amsmath,amsthm,textcomp}
\usepackage{microtype}
\usepackage[lined,boxed,commentsnumbered,ruled,vlined,noend,linesnumbered,boxed]{algorithm2e}
\usepackage{authblk}
\usepackage{enumitem} 
\usepackage[bookmarks=true,pdfborder={0 0 0}]{hyperref}

\title{Geometric  Dominating-Set and Set-Cover via Local-Search}
\author[1]{Minati De\footnote{Partially supported by DST-INSPIRE Faculty grant DST-IFA-14-ENG-75,   IIT Delhi New Faculty SEED Grant  NPN5R and SERB-MATRICS grant MTR/2021/000584.}}
\affil{Deptartment of Mathematics\\ Indian Institute of  
Technology Delhi, India\\
\texttt{minati@maths.iitd.ac.in}
}
\author[2]{Abhiruk Lahiri}

\affil{Department of Computer Science and Automation\\ Indian Institute of Science, Bangalore, India\\
\texttt{abhiruk@iisc.ac.in}
}
\date{}

\usepackage{wrapfig}
\usepackage{IEEEtrantools}
\usepackage{color}

\makeatletter
\DeclareFontFamily{U}{tipa}{}
\DeclareFontShape{U}{tipa}{m}{n}{<->tipa10}{}
\newcommand{\arc@char}{{\usefont{U}{tipa}{m}{n}\symbol{62}}}%

\newcommand{\arc}[1]{\mathpalette\arc@arc{#1}}

\newcommand{\arc@arc}[2]{%
  \sbox0{$\m@th#1#2$}%
  \vbox{
    \hbox{\resizebox{\wd0}{\height}{\arc@char}}
    \nointerlineskip
    \box0
  }%
}
\makeatother

\makeatletter
\@addtoreset{footnote}{page}
\makeatother
\renewcommand{\thefootnote}{\ifcase\value{footnote}\or(*)\or
(**)\or(***)\or(****)\or(\#)\or(\#\#)\or(\#\#\#)\or(\#\#\#\#)\or($\infty$)\fi}

\newcommand{\IR}{\mathbb{R}}

\DeclareMathOperator{\interior}{int}
\DeclareMathOperator{\petal}{Petal}
\DeclareMathOperator{\NCpetal}{NCpetal}
\DeclareMathOperator{\CF}{CF}
\DeclareMathOperator{\NVD}{NVD}
\DeclareMathOperator{\Cell}{Cell}

\usepackage{chngcntr}
\counterwithout{paragraph}{subsubsection}

\newtheorem{theorem}{Theorem}
\newtheorem{lemma}{Lemma}

\newtheorem{definition}{Definition}

\newtheorem{property}{Property}
\newtheorem{invariant}{Invariant}
\newtheorem{claim}{Claim}

\everymath{\displaystyle}

\begin{document}
\maketitle
\begin{abstract}
In this paper, we study two classic optimization  problems: minimum geometric dominating 
set and set 
cover. In the dominating-set problem, for  a given set of objects in {the} plane as input, the 
objective is to choose a  minimum number of input objects such 
that every input object is dominated by the chosen set of objects. Here, one 
object is dominated by  {another} if both of them have {a} nonempty intersection 
region. 
For the second problem, for a given  set of points and  
objects {in a plane}, the objective is to choose {a} minimum number of objects
 to cover all the points. This is a special version of the set-cover problem. 

{Both problems have been well studied subject to various restrictions on the input objects.}
These problems {are} 
APX-hard {for object sets} consisting of axis-parallel rectangles, ellipses, $\alpha$-fat 
objects 
of constant description complexity, and convex polygons. On the other hand, 
PTASs (polynomial time approximation schemes) are known  {for object sets} consisting of disks or unit squares. 
Surprisingly, {a} PTAS was unknown even 
for arbitrary squares. 

 For both  problems 
obtaining a PTAS remains open for a large class of objects. 

For the dominating-set problem, we prove that {a} popular local-search 
algorithm leads to an  $(1+\varepsilon)$ approximation for object sets consisting of homothetic set of convex objects (which includes   arbitrary squares, $k$-regular polygons, translated and scaled copies of a convex set, etc.) in $n^{O(1/\varepsilon^2)}$ time. 
On the other 
hand, the same technique leads to a PTAS  for geometric covering problem when 
the 
objects are  convex pseudodisks (which includes disks, unit height rectangles, homothetic convex objects, etc.). 
As a consequence, we obtain an  easy to implement approximation algorithm for both problems for a large class of objects,  significantly improving the best 
known approximation guarantees. 

\end{abstract}

\section{Introduction}

\subsection{Problems Studied}
We consider two fundamental combinatorial optimization problems in a geometric context, dominating-set and set-cover. Let $\cal P$  be a subset of the real plane $\mathbb{R}^2$, and let $\mathscr{S}$ be a collection of subsets of $\cal P$, called \emph{objects}. A subset $\mathscr{S}' \subseteq \mathscr{S}$ is a \emph{dominating-set} if every element of $\mathscr{S}$ has a nonempty intersection with at least one element of $\mathscr{S}'$. A subset $\mathscr{S}'' \subseteq \mathscr{S}$ is a \emph{cover} if every point of $\cal P$ lies within at least one element of $\mathscr{S}''$. The \emph{dominating-set} and \emph{set-cover} problems involve computing a minimum cardinality dominating-set and set-cover, respectively. Both problems have a wealth of theoretical results and practical applications. 
Geometric set-cover 
problem has many application in real world for example wireless sensor 
networks, optimizing number of stops in an existing transportation network, job 
scheduling~\cite{BansalP14,ClarksonV07,HarPel12}.

\subsection{Local Search}
{It is well known that both of these problems are NP-hard in the most general setting, and hence researchers have focused on approximation algorithms. In this paper, we analyze an approach based on local search.}
Local search is a popular heuristic algorithm. This is an iterative  algorithm 
which starts with a feasible solution and improves the solution after each 
iteration until a locally optimal solution is reached. One big advantage of 
local search  is that it is very easy to implement and  easy to 
parallelize~\cite{Cohen-AddadM15}. As mentioned by Cohen-Addad and 
Mathieu~\cite{Cohen-AddadM15}, it is interesting to analyze such algorithms 
even when  alternative, theoretically optimal polynomial-time algorithms are 
known. 

\subsection{Our Results} 
{Our results on the dominating-set problem apply under the assumption that the input consists of homothets of a convex body in the plane, that is, the elements of $\mathscr{S}$ are equal to each other up to translation and positive uniform scaling. This includes a large class of natural object sets, such as collections of squares of arbitrary size, collections of regular $k$-gons of arbitrary size, and collections of circular disks of arbitrary radii.}
First, we show that   the standard local search 
algorithm leads to  a  polynomial time approximation scheme (PTAS)  for 
 computing  a minimum  dominating-set 
 of homothetic convex objects. 
For the analysis, we  
use {a} separator-based  technique, which was introduced independently by Chan {and}
Har-Peled~\cite{ChanH09} and  Mustafa  {and} Ray~\cite{MustafaR10}. The main 
part of this proof technique  is to show the existence of a planar graph 
satisfying a \emph{locality 
condition} {(to be defined in Section~\ref{anls})}. Gibson et al. \cite{GibsonP10} used the same paradigm where the 
objects were arbitrary disks. Inspired by their work, we ask whether we can generalize their framework to more general objects.
Our result on {the dominating-set problem} can be viewed as a non-trivial generalization of 
their result. 
To show the planarity, first, we decompose (or shrink) a set of homothetic 
convex 
objects (which are returned by the optimum algorithm and the local search 
algorithm) into a set of interior disjoint  objects so that each input 
object has a ``trace'' in this new set of objects. This
decomposition is motivated from the idea 
of core decomposition  introduced by Mustafa et al.~\cite{Ray}, and 
this technique could be of independent interest. 
Next, we   consider the nearest-site Voronoi diagram for this set of  
disjoint  objects with respect to the well-known convex distance 
function. The decomposition  ensures that each site has a nonempty cell in the  Voronoi diagram. Finally, we 
show that the dual of this Voronoi diagram satisfies the  {locality 
condition}.  Note that if  homothets of a centrally symmetric convex object are given, then one can avoid the disjoint decomposition, and the analysis is much simpler.

{Our results on the set-cover problem apply under the assumption that the input consists of a collection of  convex pseudodisks in the plane.}  A set of objects is said to be a collection of \emph{pseudodisks}, if the boundaries of every pair of them intersect at most twice. Note that this generalizes collections of  homothets. We use a 
similar technique as the previous one. First, we show that we can decompose (or 
shrink) a set of  pseudodisks (which are returned by the optimum algorithm and the local search 
algorithm) into a set of interior disjoint  objects so that each input 
point has 
a ``trace'' in this new set of objects.
We consider a graph ${\cal G}$ in which each vertex corresponds to a {shrunken} 
object, and {two vertices are joined by an} edge if the corresponding objects share an edge 
in their boundary. Since the shrunken objects are interior disjoint with each 
other, the graph ${\cal G}$ is planar. We prove that the graph ${\cal G}$ 
satisfies the locality condition.

{Given $\varepsilon > 0$, a \emph{$(1+\varepsilon)$-approximation algorithm} for the dominating-set (resp., set-cover) problem returns a dominating-set (resp., set-cover) whose cardinality is larger than the optimum by a factor of at most $(1+\varepsilon)$. Our results are given below.}

\begin{theorem}\label{Thm:DominatingSet}
{Given a set $\mathscr{S}$ of n convex homothets in $\IR^2$ and $\varepsilon > 0$, there exists a $(1+\varepsilon)$ approximation algorithm for dominated set based on local search that runs in time $n^{O(1/\varepsilon^2)}$.}
\end{theorem} 
\begin{theorem}\label{Thm:SetCover}
{Given a set $\mathscr{S}$ of n convex pseudodisks in $\IR^2$  and $\varepsilon > 0$, there exists a $(1+\varepsilon)$ approximation algorithm for set-cover based on local search that runs in time $n^{O(1/\varepsilon^2)}$.}
\end{theorem} 

\subsection{Related Works}

Our work is motivated by recent progress on approximability of various 
fundamental geometric optimization problems like finding maximum  independent 
sets~\cite{AdamaszekW14}, 
 minimum hitting set of geometric intersection graphs~\cite{MustafaR10}, and 
minimum geometric 
set 
covers~\cite{Ray}. 

{\bf Dominating-Set:} 
The minimum dominating-set problem is 
NP-complete for general graphs~\cite{GareyJ79}. From the result of Raz and 
Safra~\cite{RazS97}, it follows that it is NP-hard even to obtain {a} $(c\log 
\Delta)$-approximate dominating-set for general graphs, where 
$\Delta$ is the maximum degree of a node in the graph and $c\,(> 0)$ is any 
constant 
(see \cite{LenzenW10}).

Researchers have studied the problem for different graph 
classes like planar graphs, intersection graphs, bounded arboricity graphs, etc.
Recently, Har-Peled and Quanrud~\cite{Har-PeledQ15} proved that   local search produces
 a PTAS for graphs with polynomially 
bounded expansion. 
Gibson and Pirwani~\cite{GibsonP10} gave a

PTAS for the intersection graphs of arbitrary disks.
Unless $P=NP$~\cite{DinurS14}\footnote{Originally the assumption was 
$NP\nsubseteq DTIME(n^{O(\log\log n)})$. This  assumption was improved to 
$P\neq NP$ recently by Dinur and Steurer~\cite{DinurS14}.}, it is not possible 
{to compute a} 
$((1-\epsilon)\ln n)$-approximate dominating-set  in 
polynomial time  for $n$  homothetic 
polygons~\cite{Feige98, Viggo-Kann-Thesis,Leeuwen-Thesis}.
Erlebach and van Leeuwen \cite{ErlebachL08} proved that  the problem is 
APX-hard for the 
intersection graphs of axis-parallel rectangles, ellipses,  $\alpha$-fat 
objects of constant description complexity, and of convex polygons with 
$r$-corners ($r\geq 4$),  i.e., there is no PTAS for these  unless 
$P=NP$.

{Effort has been devoted to related problems involving various objects such as
squares, regular polygons, etc..}
Marx~\cite{Marx06} proved that the problem is $W[1]$-hard for unit 
squares, which implies that no 
efficient-polynomial-time-approximation-scheme (EPTAS) is possible unless 
$FPT=W[1]$ 
\cite{Marx08}. 
 The best known approximation factor for homothetic $2k$-regular 
polygons is $O(k)$ due to Erlebach and van Leeuwen \cite{ErlebachL08}, where 
$k>0$. They also obtained an $O(k^2)$-approximation algorithm for homothetic 
$(2k+1)$-regular polygons.
Even worse, for  the homothetic convex polygons where each 
polygons has $k$-corners, the best known result is $O(k^4)$-approximation. 
{Currently, there is no}  PTAS  even for arbitrary squares.
We consider the problem for a set of homothetic convex objects.

{\bf Set-Cover:}
The set-cover problem is known to be NP-complete~\cite{Karp}.
The geometric variant 
{has received a great amount of attention} due to its wide applications (for example the recent 
breakthrough of Bansal and Pruhs~\cite{BansalP14}). 
Unfortunately, the
geometric version of the problem also remains NP-complete even when the objects 
are unit disks or unit squares~\cite{ChanG, HochbaumM87}. 

Erlebach and van Leeuwen~\cite{ErlebachL10} obtained a PTAS for the geometric  
set-cover problem when the objects are unit squares. Recently, Chan and 
Grant~\cite{ChanG} showed that the  problem is APX-hard when the 
objects are  axis-aligned  rectangles.  They  extended the results to several 
other 
classes of objects including axis-aligned ellipses in $\IR^2$, axis-aligned 
slabs, 
downward shadows of line segments,  unit balls in $\IR^3$, axis-aligned 
cubes in $\IR^3$. A QPTAS {was developed by} Mustafa et. al.~\cite{Ray}  for the problem when the objects are 
pseudodisks. The current {state of the art} lacks a PTAS  when the objects are 
pseudodisks which includes  a large class of objects: arbitrary squares, 
arbitrary regular polygons, homothetic convex objects.

In the weighted setting, 
Varadarajan 
introduced the idea of quasi-uniform 
sampling to obtain {an} $O(\log \phi(OPT))$-approximation guarantees in the weighted 
setting 
for 
a large class of objects for which such guarantees were known in the 
{unweighted} case \cite{Varada}.  Here  $\phi(OPT)$ is the union complexity of 
the objects in the optimum set $OPT$. Very recently, Li and Jin  proposed a 
PTAS   for {the}
weighted version of the problem when the objects are unit disks \cite{LiJ}.

In \cite{HarPel12}, {the} authors  described a 
PTAS for the problem {of} computing a minimum cover of given points by a set of 
weighted fat objects, by allowing them to expand by some $\delta$-fraction.
  A 
{multi-cover variant of the problem (where each point is covered by at least k sets)} under 
geometric settings was studied in \cite{Chekuri}.

\subsection{Organization}

In Section~\ref{S-Alg}, we present a general  
algorithm based on the local search technique. For the sake of 
completeness, we present a high-level view of the analysis technique of local search
which was introduced by Chan \& 
Har-Peled~\cite{ChanH09} and Mustafa \& Ray~\cite{MustafaR10}.
In Section~\ref{tools}, we prove two results for a set of  pseudodisks 
which 
are common tools for analyzing both dominating-set and geometric set-cover 
problem. 
Thereafter, in 
Section~\ref{convPoly} and Section~\ref{Appendix} we prove    the locality condition 
for the dominating-set prolem when the objects are homothets of a convex polygon and of a centrally symmetric convex polygon, respectively.  In Section~\ref{setCover}, we prove the locality condition for {the} geometric set-cover problem when the  objects are convex  pseudodisks.

\subsection{{Notation and Preliminaries}}
Throughout the paper, we use capital  {letters to denote objects} and
 caligraphic font to denote   {sets of objects}.  {We make the general-position assumption that if two objects of the input set have a nonempty intersection, then their interiors intersect. No three object boundaries intersect in a common point.}
We denote the set $\{1,2,\ldots, n\}$ as $[n]$.
By a \emph{geometric object} (or object, in short) $R$, 
we 
refer  to a simply connected compact region  in $\IR^2$ with nonempty interior. 
In other words, 
the object $R$ is a closed region bounded by a closed 
Jordan curve $\partial{R}$. The $\interior(R)$ is defined as all the 
points in $R$ which do not appear in the boundary  $\partial{R}$. 
Given two objects $U$ and $V$, we say that $U$ has an \emph{interior overlap} with $V$ if $\interior(U) \cap \interior(V) \neq \emptyset$, and given a set of objects $\mathcal{V}$, we say that $U$ has an \emph{interior overlap} with $\mathcal{V}$ if $U$ has an interior overlap with any $V \in \mathcal{V}$.

For a   set of objects $\cal R$,  we define the 
\emph{cover-free region} of any object $R_i\in {\cal R}$  as $\CF(R_i, {\cal R})=  \bigcap_{\substack{R_j\in {\cal R}\\ R_j \neq R_i}} R_i \setminus R_j$. Note that $\CF(R_i,{\cal R})\cap R_j=\emptyset$ for all 
$R_i,R_j(i\neq j)\in {\cal R}$. When the  {underlying} set of objects $\cal R$ is obvious, we use the term 
$\CF(R_i)$ 
instead of $\CF(R_i, {\cal R})$.
A collection of geometric objects $\cal R$
is said to form a family of \emph{pseudodisks} if the boundary of any two 
objects cross each other at most twice.  A collection of  geometric 
objects $\cal R$ is said to be \emph{cover-free}
 if no object $R\in {\cal R}$ is covered by the union of the objects in ${\cal 
R}\setminus R$, in other words, $\CF(R, {\cal R})\neq \emptyset$ for all  
objects in $\cal R$. 
  Two objects  are \emph{homothetic} to each other if one object can be 
obtained from the other by scaling and translating.

 Consider the \emph{convex distance function} with respect to a convex object $C$ with a fixed interior point as \emph{center} as follows.

\begin{definition}\label{def_1_cd}
 Given $p_1, p_2 \in \mathbb{R}^2$, \emph{convex distance function} induced by $C$, denoted by $\delta_C(p_1,p_2)$, is the smallest $\alpha \geq 0$ such that $p_1, p_2 \in \alpha C$ while  the center of $C$ is at $p_1$. 
\end{definition}
 It was first introduced by Minkowski in 1911~\cite{KelleyN, ChewD85}. Note that this function satisfies the following  properties.

\begin{property}
\label{prop:convex_dist}
\begin{enumerate}[label=(\roman*)]
\item The function  $\delta_C$ is symmetric (i.e., $\delta_C(p_1,p_2)=\delta_C(p_2,p_1)$) if and only if $C$ is centrally symmetric.
\item Let  $p_1$ and $p_3$ be any two points in $\IR^2$ and let $p_2$ be 
any point 
on the line segment $\overline{p_1p_3}$, then   $\delta_C(p_1,p_3)= 
\delta_C(p_1,p_2)+\delta_C(p_2,p_3)$.
\item The distance function $\delta_C$ follows 
the triangular inequality, i.e.,  and $\delta_C(p_1,p_3)\leq 
\delta_C(p_1,p_2)+\delta_C(p_2,p_3)$, where  $p_1$, $p_2$ and 
$p_3$  are any three points in $\IR^2$.
\end{enumerate}
\end{property}

\section{Local-Search Algorithm}
\label{S-Alg}
We use a  standard local search algorithm~\cite{MustafaR10} as 
given in Algorithm~\ref{ALG-Code}.

\begin{algorithm}[h]
 \SetAlgoLined
\small
\KwIn{A set of $n$ objects $\mathscr{S}$ in $\IR^2$  { and a parameter $b$}}
Initialize ${\cal A}$ to an arbitrary  subset of $\mathscr{S}$ which is a 
feasible solution\;
\While{$\exists$ ${\cal X}\subseteq {\cal A}$ of size at 
most $b$, and 
 ${\cal X}'\subseteq \mathscr{S}$ of size at most $|{\cal X}|-1$ such 
that 
$({\cal A}\setminus {\cal X})\cup 
{\cal X}'$ is a feasible solution}
{
set ${\cal A}\leftarrow 
({\cal A}\setminus {\cal X})\cup 
{\cal X}'$\;
}
Report ${\cal A}$\;
\normalsize
 \caption{Local-Search($\mathscr{S},b$)}
\label{ALG-Code}
\end{algorithm}

A subset of objects ${\cal A}\subseteq \mathscr{S}$  is referred to  
{\it $b$-locally 
optimal} if one 
cannot obtain a smaller feasible solution  by removing a subset ${\cal 
X}\subseteq 
{\cal A}$ of 
size at most $b$ from ${\cal A}$ and replacing that with a subset of size at 
most 
$|{\cal X}|-1$ from $\mathscr{S}\setminus {\cal A}$.  Our algorithm  computes 
a 
$b$-locally optimal 
set of objects for $b=\frac{\alpha}{\epsilon^2}$, where $\alpha>0$ is a 
suitably large constant. Observe that at the end  of the while-loop, the set  $\cal A$ is $b$-locally 
optimal, and the set  $\cal A$ is 
cover-free. 

Since the size of ${\cal A}$ is decreased by at least one after each update in 
Line 3, 
 the number of iterations of the  while-loop is at most $n$, and
 each iteration takes $O(n^b)$ time as it needs to check every subset of size 
at most $b$.  So, this while-loop needs  $O(n^{b+1})$ time.  
Thus, total time complexity of  
the above algorithm is 
$O(n^{b+1})$.

\subsection{Analysis of Approximation}\label{anls}
 {We will be analyzing the algorithm's performance with respect to both problems. When there is a difference, we will indicate the specific context within which the analysis is being performed (set-cover or dominating-set). }
Let ${\cal O}$ be the optimal solution  and ${\cal A}$ 
be 
the solution 
returned by our local search algorithm. 
Note that both ${\cal O}$ and ${\cal A}$  ensure the following. 
\begin{claim}\label{cD0.0}
  {For any object $A\in  {\cal A}$ (resp., $O\in {\cal  O}$),  $\CF(A, {\cal A})$ (resp.,  $\CF(O, {\cal  O})$)  is nonempty. In other words, ${\cal A}$ (resp., ${\cal  O}$) is cover-free. } 
\end{claim}

We can assume that no object $S\in \mathscr{S}$ is properly contained in any other object of $\mathscr{S}$. We can ensure this by an initial pass over the input objects in which we remove any object of the input that is contained within another object. Thus, we can assume that there is no 
object $S\in 
\mathscr{S}\setminus {\cal A}$ which completely contains any object of ${\cal A}$. 
Similarly, we can 
assume that no object in ${\cal O}$ is completely contained in any object 
from 
$\mathscr{S}\setminus {\cal O}$.   Let ${\cal A}'={\cal A}\setminus {\cal O}$, ${\cal O}'={\cal O}\setminus 
{\cal A}$. 

In the context of the dominating-set problem,
let  $\mathscr{S}'\subset \mathscr{S}$ be the set containing  all objects of $\mathscr{S}$ which
are not dominated by any 
object in ${\cal A}\cap {\cal O}$. 
   Note that there does not exist an object $O\in {\cal O}'$ which covers $\CF(A_1, {\cal A}')\cup \CF(A_2, {\cal A}')$, $A_1,A_2\in {\cal A}'$, otherwise local search would replace  $A_1$ and $A_2$ by $O$. Similarly, there does not exist an object $A\in {\cal A}'$ which covers $\CF(O_1, {\cal O}')\cup \CF(O_2, {\cal O}')$, $O_1,O_2\in {\cal A}'$ otherwise it would contradict the optimality of ${\cal O}$.

Now we are going to eliminate {the} same number of objects from both ${\cal A}'$ and 
${\cal O}'$ to ensure that for any $A\in  {\cal A}'$,  $\CF(A, {\cal A}')$ is not properly
  contained in any object in ${\cal  O}'$.
 Let $O\in  {\cal O}'$ be an object that properly contains $\CF(A, {\cal A}')$ for an object $A\in {\cal A}'$. Let  $\mathscr{S}''$ be the the set containing  all objects of $\mathscr{S}'$ which
are not dominated by $O$. Note that both the sets ${\cal A}'\setminus A$ and ${\cal O}'\setminus O$ dominates $\mathscr{S}''$.
We reset $\mathscr{S}'\gets \mathscr{S}''$.
   We remove $A$ and $O$ from ${\cal A}'$ and ${\cal O}'$, respectively by  updating {${\cal A'}\gets{\cal A}'\setminus A$} and ${\cal O}'\gets {\cal O}'\setminus O$. We repeat this until there does not exist any object $O\in \cal{O}'$ that properly contains an object  $A\in {\cal A}'$. 
  
Similarly, if there exists an object $A\in {\cal A}'$ that properly contains $\CF(O, {\cal O}')$ for an object $O \in {\cal O}'$, we  update ${\cal A}'\gets{\cal A}'\setminus A$ and ${\cal O}'\gets{\cal O}'\setminus O$.  Let  $\mathscr{S}''$ be the the set containing  all objects of $\mathscr{S}'$ which
are not dominated by $A$. We reset $\mathscr{S}'\gets\mathscr{S}''$. We repeat this until there does not exist any object $A\in {\cal A}'$ that properly contains $\CF(O, {\cal O}')$ for an object $O \in {\cal O}'$. This ensures the following. 

\begin{claim}\label{cD0.1}
 For any object $A\in  {\cal A}'$ (resp., $O\in {\cal O}'$),  $\CF(A, {\cal A}')$ (resp.,  $\CF(O, {\cal O}')$)  is not properly
  contained in any object in ${\cal O}'$  (resp.,  $ {\cal A}'$).
\end{claim}

Observe that $|\cal{O}\setminus \cal{O}'|=|\cal{A}\setminus \cal{A}'|$.
Finally, we  will show that  $|{\cal A}'|\leq (1+\epsilon) |{\cal O}'|$ which implies 
that $|{\cal A}|\leq (1+\epsilon) |{\cal O}|$. 

In the context of geometric covering,  we do the similar process as discussed above to ensure Claim~\ref{cD0.1}.  Here, let ${\cal P}'$ be the set containing all 
points of $\cal P$ which  are
covered by
object in ${\cal A}'\cap {\cal O}'$.

{Henceforth, ${\cal A}',{\cal O}', {\cal P}'$ and $\mathscr{S}'$ will be denoted as ${\cal A},{\cal O}, {\cal P}$ and $\mathscr{S}$, respectively, satisfying both Claim~\ref{cD0.0} and \ref{cD0.1}.}

In Sections~\ref{loc-cond} and~\ref{setCover}, we prove \emph{locality conditions} for the dominating-set and set-cover problems, respectively. These conditions are presented in Lemmas \ref{lem_loc_dom} and \ref{lem_loc_cov}, respectively.

\begin{lemma}[Locality Condition for Dominating-Set]\label{lem_loc_dom}
 There exists a planar graph $ {\cal G}=({\cal A}\cup {\cal O}, {\cal E})$ 
such that for all 
$S\in \mathscr{S}$, if $S$ is dominated by at least one object of $\cal A$ and at least one object of $\cal O$, then there exists $A\in {\cal A}$ and $O\in {\cal O}$ both of which dominate $S$ and $(A,O)\in {\cal E}$.
\end{lemma}

\begin{lemma}[Locality Condition for Set-Cover]
\label{lem_loc_cov}
 There exists a planar graph $ {\cal G}=({\cal A}\cup {\cal O}, {\cal E})$ 
such that for all points
$p\in {\cal P}$, if $p$ is covered by at least one object  of $\cal A$ and at least one object of $\cal O$, then there exists $A\in {\cal A}$ and $O\in {\cal O}$ both of which cover $p$ and $(A,O)\in {\cal E}$.  
\end{lemma}

Once we have established both of these
locality condition lemmas, the analysis of the 
algorithm is same 
as in~\cite{MustafaR10}. 
For the sake of completeness, we provide the following analysis.
As  the graph ${\cal G}$ is planar, the following planar separator theorem can 
be used.

\begin{theorem}[Frederickson~\cite{Frederickson87}]\label{THseparator}
 For any planar graph ${\cal G}=({\cal V},{\cal E})$  {with} $n$ vertices  {and a parameter $1\leq r 
\leq n$},  there is 
a set ${\cal X}\subseteq 
{\cal V}$ of size at most $\frac{c_1n}{\sqrt{r}}$, such that  ${\cal V}\setminus 
{\cal X}$ can be 
partitioned into  {$\lceil n/r \rceil$ sets ${\cal V}_1,{\cal V}_2,\ldots {\cal V}_{\lceil n/r \rceil}$} 
satisfying (i) 
$|{\cal V}_i|\leq c_2r$, (ii) $N({\cal V}_i)\cap {\cal V}_j=\emptyset$ for 
$i\neq j$, and 
$|N({\cal V}_i)\cap {\cal X}|\leq c_3\sqrt{r}$, where $c_1,c_2,c_3>0$ are 
constants, and $N({\cal V}')=\{U\in {\cal V}\setminus {\cal V}'\mid \exists V\in 
{\cal V}' \text{ with } (U,V)\in 
{\cal E}  \}$.  
\end{theorem}

We apply Theorem~\ref{THseparator} to the graphs described in Lemmas~\ref{lem_loc_dom} and~\ref{lem_loc_cov}, setting $r = b/c_2$, where $c_2$ is the constant of Theorem~\ref{THseparator}.   Here,  $n = |{\cal A}| + |{\cal O}|$ and $r = c_4/\epsilon^2$, for some constant $c_4$. So, $|{\cal 
V}_i| \leq 
b$.
Let ${\cal A}_i={\cal A}\cap {\cal V}_i$ and ${\cal O}_i={\cal O}\cap 
{\cal 
V}_i$.  Note that we must have 
\begin{align}\label{eq0}
|{\cal A}_i|\leq |{\cal O}_i|+|N({\cal V}_i)\cap {\cal X}|,
\end{align}

otherwise our local search would  {continue to} replace ${\cal A}_i$ by 
${\cal O}_i\cup N({\cal V}_i)$ {, resulting in a better solution}.  For a suitable constant $c_5$, we now have
{\small
\begin{IEEEeqnarray*}{rcl's}
|{\cal A}|
& \leq & |{\cal X}|+ \sum\limits_{i}|{\cal A}_i| & (Each element of ${\cal 
Q}$ 
either belongs to 
${\cal A}_i$ or ${\cal X}$)\\
& \leq & |{\cal X}| + \sum\limits_{i}|{\cal O}_i|+ \sum\limits_{i}|N({\cal 
V}_i)\cap {\cal X}| & (Follows from 
Equation~\ref{eq0}) \\
& \leq & |{\cal O}|  + |{\cal X}| +  \sum\limits_{i}|N({\cal V}_i)\cap {\cal X}| & 
(${\cal 
O}_i$ are disjoint subsets 
of ${\cal O}$)\\
& \leq & |{\cal O}| + \frac{c_5(|{\cal A}|+|{\cal O}|)}{\sqrt{b}} & 
($ \sum\limits_{i}|N({\cal V}_i)\cap {\cal X}|\leq \lceil {n/r} \rceil( c_3 \sqrt{r})$ and $|{\cal X}|\leq c_1 (|{\cal A}|+ |{\cal O}|)/ \sqrt{r}$ )\\
|{\cal A}| & \leq & \frac{1+c_5/\sqrt{b}}{1-c_5/\sqrt{b}}|{\cal O}| & (By 
rearranging)\\
|{\cal A}| & \leq & (1+\epsilon)|{\cal O}| & ($b$ is large enough constant 
times 
$\frac{1}{\epsilon^2}$). 
\end{IEEEeqnarray*}
}

\section{Tools for Constructing Disjoint  Objects}\label{tools}
In this section, we present two tools (or Lemmata) which are {essential} for 
analyzing our main results. 
{An important step in our analysis (and particularly in the construction of the planar graph of Section~\ref{anls}) involves replacing a collection of overlapping objects that cover a given region with a collection of non-overlapping objects that cover the same region. This leads to the notion of a {\it decomposition}.}
The {decomposition}, we define here, is inspired by the idea of  
core decomposition introduced by Mustafa et al.~\cite{Ray}. 
\begin{definition}\label{cdDef}
{Given a set of convex objects ${\cal R} = \{R_1, \ldots, R_n\}$, a set $\widetilde{\cal R}=\{ {{\widetilde{R}}_1},\ldots,{{\widetilde{R}}_n}\}$ of  convex objects is called a \emph{sub-decomposition} if for each $i \in [n]$, ${{\widetilde{R}}_i}\subseteq R_i$.
Such a set $\widetilde{\cal R}$ is called a \emph{decomposition} if the same region is covered, that is, $\bigcup_{i \in [n]}  {{\widetilde{R}}_i} = \bigcup_{i \in [n]} R_i$. We refer ${\widetilde{R}}_i$ as the \emph{trace} of $R_i$, $i \in [n]$. Further, if the elements of $\widetilde{\cal R}$ have pairwise disjoint interiors, the decomposition/sub-decomposition is said to be \emph{disjoint}.}
\end{definition}

First, we prove the 
following lemma which is a reminiscent of \cite[Lem 3.3]{Ray}.  {Edelsbrunner~\cite{Edelsbrunner95} introduced a very similar decomposition in the context of Euclidean disks.}


\begin{lemma} \label{lDecompose1}
For a cover-free set  of convex pseudodisks ${\cal R}=\{ 
R_1,\ldots,R_n\}$, 
there exist a disjoint decomposition $\widetilde{\cal R}=\{ 
{{\widetilde{R}}_1},\ldots,{{\widetilde{R}}_n}\}$ such that $\CF(R_j, {\cal R})\subseteq 
{{\widetilde{R}}_j}$, for all $j\in [n]$.
\end{lemma}

\begin{proof}
The proof is constructive.
The algorithm to construct a disjoint decomposition    $\widetilde{\cal 
R}=\{ 
{{\widetilde{R}}_1},\ldots,{{\widetilde{R}}_n}\}$  of ${\cal R}=\{ R_1,\ldots,R_n\}$ 
is as  follows. This is an $n$-phase algorithm.
After the $i^{th}$ phase, the  following invariants are maintained, for all  
$i\in [n]$.

\begin{invariant}\label{inv1.1}
The objects in 
$\widetilde{\cal R}^{i}=\{{{\widetilde{R}}_1}^i,\ldots,{{\widetilde{R}}_n}^i\}$ form a 
decomposition of ${\cal R}=\{ 
R_1,\ldots,R_n\}$ such that (i) $\CF(R_j)\subseteq {{\widetilde{R}}_j}^i$ for all 
$j\in [n]$, and  (ii) 
$\interior({{\widetilde{R}}_t}^i) \cap \interior({{\widetilde{R}}_q}^i)=\emptyset$ 
where $t\neq q$ and $1\leq 
t \leq i$, $1\leq q \leq n$. 

\end{invariant}

\begin{invariant}\label{inv1.2}
The objects in $\widetilde{\cal 
R}^{i}=\{{{\widetilde{R}}_1}^i,\ldots,{{\widetilde{R}}_n}^i\}$ form a {collection of} convex pseudodisks. 
\end{invariant}

We initialize $\widetilde{\cal R}^0={\cal R}$. This satisfies both 
invariants.
At the beginning of the
$i^{th}$ phase,  we set $X=\widetilde{ 
R}_i^{i-1}$. 
Let ${\cal 
R}_{\pi}^i=\{{{\widetilde{R}}_{\pi(1)}}^{i-1}, \ldots, 
{{\widetilde{R}}_{\pi(\ell)}}^{i-1}\}$, $0\leq \ell<n$
be the set of  objects in  $\widetilde{\cal 
R}^{i-1}$ 
that intersect  $\interior({{\widetilde{R}}_i}^{i-1})$. In other words,
$\interior({{\widetilde{R}}_i}^{i-1})\cap 
\interior({{\widetilde{R}}_{\pi(j)}}^{i-1})\neq \emptyset$ for any  $\pi(j)\in \Pi$, 
where $\Pi=\{\pi(1), \ldots, \pi(\ell)\}$.

Consider any  object ${{\widetilde{R}}_{\pi(j)}}^{i-1}\in {\cal R}_{\pi}^i$. 
As ${{\widetilde{R}}_{\pi(j)}}^{i-1}$  and $X$ are pseudodisks, 
their respective boundaries intersect in two points.
 Let 
$p_1$ and $p_2$ be these  two intersection points. 
By convexity, the line segment $\overline{p_1p_2}$  is contained in both
 ${{\widetilde{R}}_{\pi(j)}}^{i-1}$  and $X$.
Let ${\cal C}_1$ (respectively, 
${\cal C}_2$) 
be the part of the boundary of ${{\widetilde{R}}_{\pi(j)}}^{i-1}$ (respectively, 
$X$) 
that lie
inside $X$ (respectively, ${{\widetilde{R}}_{\pi(j)}}^{i-1}$).
 We replace  both ${\cal C}_1$ and 
${\cal C}_2$ by the line segment $\overline{p_1p_2}$. In this way, we obtain new  convex objects 
${{\widetilde{R}}_{\pi(j)}}^{i}\subseteq {{\widetilde{R}}_{\pi(j)}}^{i-1}$ and $X_j\subseteq X$  that have interiors that are  pairwise disjoint  with each other, and 
${{\widetilde{R}}_{\pi(j)}}^{i}\cup X_j={{\widetilde{R}}_{\pi(j)}}^{i-1}\cup X$.  See 
Figure~\ref{figDecompose2} for illustration. 

\begin{figure}[!htb]
    \centering
    \begin{minipage}{.5\textwidth}
        \centering
        \includegraphics[width=0.5\linewidth, height=0.15\textheight,page=1]{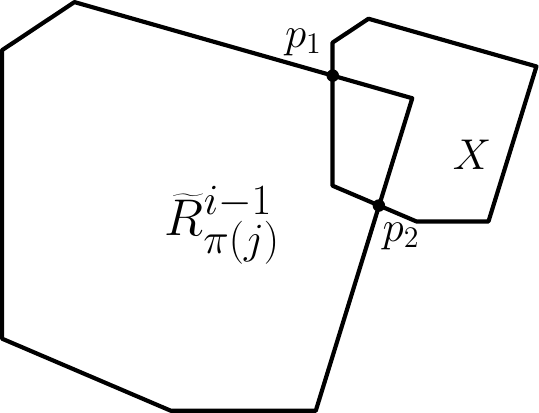}
      
    \end{minipage}%
    \begin{minipage}{0.5\textwidth}
        \centering
        \includegraphics[width=0.5\linewidth, height=0.15\textheight,page=2]{u-decompose1.pdf}
    \end{minipage}
     \caption{{Illustration of Lemma~\ref{lDecompose1}}.}
      \label{figDecompose2}
\end{figure}

For all $\pi(j)\in \Pi$, we 
construct the corresponding ${{\widetilde{R}}_{\pi(j)}}^{i}$ as above. 
At the end of this phase, we assign ${{\widetilde{R}}_i}^{i}=\bigcap_{j\in \Pi} X_j$. Note  that ${{\widetilde{R}}_i}^{i}$ is also convex as it is intersection of some convex objects.
We  set  ${{\widetilde{R}}_j}^{i}={{\widetilde{R}}_j}^{i-1}$ for 
all $j(\neq i)\in [n] \setminus \Pi$.  As a result, we obtain a collection of convex objects  $\widetilde{\cal R}^{i}$.

Observe that, for any point $p$ that is contained in the union of ${\cal 
R}_{\pi}^i$, either there exists a $j$ such that this point lies within ${{\widetilde{R}}_{\pi(j)}}^{i}$, and so is covered by this set, or it lies within $X_j$ for all $j$, and hence it lies within their common intersection, which is $X$. So, $\widetilde{\cal R}^i$  is a decomposition of $\widetilde{\cal R}^{i-1}$.

Thus, after the 
$i^{th}$ phase, 
we find a decomposition $\widetilde{\cal R}^i$  such that 
$\interior({{\widetilde{R}}_i}^i)\cap \interior({{\widetilde{R}}_j}^i)=\emptyset$ for 
all 
$j(\neq i)\in \{1,\ldots, n\}$. On the other hand, we have 
$\interior({{\widetilde{R}}_t}^{i-1}) \cap  
\interior({{\widetilde{R}}_q}^{i-1})=\emptyset$ where $t\neq q$ and $1\leq t \leq 
i-1$, $1\leq 
q 
\leq n$. Combining these,  we obtain  $\interior({{\widetilde{R}}_t}^i) 
\cap \interior(  
{{\widetilde{R}}_q}^i)=\emptyset$ where $t\neq q$ and $1\leq t \leq i$, $1\leq q 
\leq n$.

Since the union of objects 
in  $\widetilde{\cal R}^{i}$  is same as the 
union of the  objects in  $\widetilde{\cal R}^{i-1}$, and the 
objects in $\widetilde{\cal R}^{i-1}$ are cover-free,  so each object  
${{\widetilde{R}}_j}^{i}$
has its cover-free region $\CF(R_j)$  which is not covered by others, for all $j\in [n]$.  Thus, 
Invariant~\ref{inv1.1} is maintained.
Now, we prove that Invariant~\ref{inv1.2} is also 
maintained.
We prove the objects in 
$\widetilde{\cal R}^i$ form pseudodisks by showing the 
following claim.

\begin{claim}\label{C1}
$\widetilde{\cal R}^i$ is a collection of convex pseudodisks.
\end{claim}

\begin{figure}
  \centering
  \begin{minipage}{.4\linewidth}
    \centering
     \subcaptionbox*{(a) Case 1}
      {\includegraphics[width=\linewidth,page=5]{u-de-cases.pdf}}

     \subcaptionbox*{(b) Case 2}
      {\includegraphics[width=\linewidth,page=1]{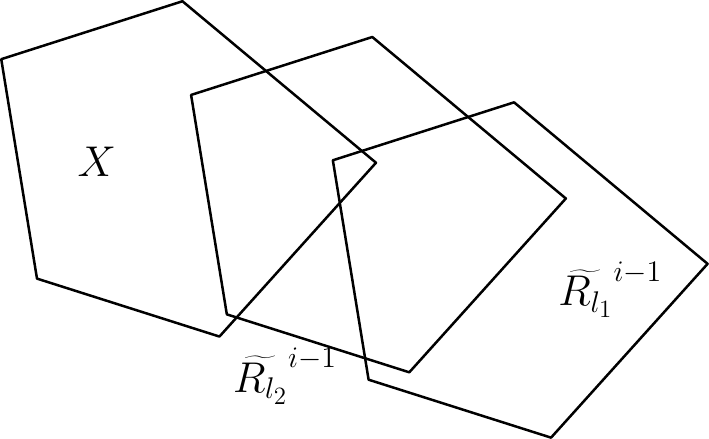}}
      
       \subcaptionbox*{(c) Case 3}
      {\includegraphics[width=\linewidth,page=3]{u-de-cases.pdf}}
\caption*{Before Phase $i$}
  \end{minipage}\quad
   \begin{minipage}{.4\linewidth}
    \centering

   \subcaptionbox*{(a) Case 1}
      {\includegraphics[width=\linewidth,page=6]{u-de-cases.pdf}}
 	  \subcaptionbox*{(b) Case 2}
      {\includegraphics[width=\linewidth,page=2]{u-de-cases.pdf}}
      
       \subcaptionbox*{(c) Case 3}
      {\includegraphics[width=\linewidth,page=4]{u-de-cases.pdf}}

    \caption*{After Phase $i$}
  \end{minipage}
 \caption{Illustration of Claim~\ref{C1}.}
 \label{figDecompose}

\end{figure}

 \begin{proof}
It suffices to show that for any two objects ${{\widetilde{R}}_{\ell_1}}^{i-1}$ 
 and ${{\widetilde{R}}_{\ell_2}}^{i-1}$ in ${R~}^{i-1}$, their boundaries $\partial{{\widetilde{R}}_{\ell_1}}^{i}$ and  $\partial {{\widetilde{R}}_{\ell_2}}^{i}$ can cross each other at most twice. 
 
 Recall the definition of $X$ from the above construction.
For any 
$R\in {\cal R}_{\pi}^i$, let $I(R)$  be the interval $R\cap \partial X$ on 
the 
boundary of $X$. Due to the Invariant~\ref{inv1.1},  no pseudodisk in $\widetilde{\cal R}^{i-1}$ is completely 
contained in another pseudodisk, so the intervals are well defined.

 There are three possible cases: 

\begin{itemize}
 \item Case 1:  $I({{\widetilde{R}}_{\ell_1}}^{i-1})\cap 
I({{\widetilde{R}}_{\ell_2}}^{i-1})=\emptyset$,
\item Case 2:  $I({{\widetilde{R}}_{\ell_1}}^{i-1})\subseteq
I({{\widetilde{R}}_{\ell_2}}^{i-1})$,
\item Case 3:  $I({{\widetilde{R}}_{\ell_1}}^{i-1})\cap 
I({{\widetilde{R}}_{\ell_2}}^{i-1})\neq \emptyset$ and 
$I({{\widetilde{R}}_{\ell_1}}^{i-1})\nsubseteq
I({{\widetilde{R}}_{\ell_2}}^{i-1})$.
\end{itemize}
In both Case 1 and Case 2 (see Figure~\ref{figDecompose}(a) and (b)),  $\partial{{\widetilde{R}}_{\ell_1}}^{i}$ and  
$\partial{{\widetilde{R}}_{\ell_2}}^{i}$ do not have any new crossing which  
$\partial{{\widetilde{R}}_{\ell_1}}^{i-1}$ and  
$\partial{{\widetilde{R}}_{\ell_2}}^{i-1}$ did not have.
 In fact they may lost intersections lying in $X$. As 
$\partial {{\widetilde{R}}_{\ell_1}}^{i-1}$ 
 and $\partial {{\widetilde{R}}_{\ell_2}}^{i-1}$  may cross each other at most 
twice, so does $\partial{{\widetilde{R}}_{\ell_1}}^{i}$ and  
$\partial{{\widetilde{R}}_{\ell_2}}^{i}$.
In Case 3 (see Figure~\ref{figDecompose}(c)),   $\partial{{\widetilde{R}}_{\ell_1}}^{i-1}$ and  
$\partial{{\widetilde{R}}_{\ell_2}}^{i-1}$ crosses each other once in $X$ and once outside $X$. The outside crossing remains same for $\partial{{\widetilde{R}}_{\ell_1}}^{i}$ and  
$\partial{{\widetilde{R}}_{\ell_2}}^{i}$, and they cross each other once along  new part of  their boundaries, i.e.,  along the boundary of $X_{\ell_1}\cap X_{\ell_2}$. Thus, the claim follows.

\end{proof} 
 
 After completion 
of the $n^{th}$ phase, we 
assign $\widetilde{\cal R}=\widetilde{\cal R}^n$.
The proof of the lemma follows from the  Invariant~\ref{inv1.1}.
\end{proof}


Now, we prove the following important lemma which we use as a tool for obtaining disjoint sub-decompositions. The previous lemma is used to obtain disjoint decomposition when the objects are pseudodisks.  When the set of objects does not satisfy the pseudodisk property, but they are shrunken  from a set of of pseudodisks, we apply the following tool to obtain a disjoint sub-decomposition.

\begin{lemma} \label{petalLemma}
Given two sets ${\cal U}$ and ${\cal V}$ of distinct convex objects such that their union forms a 
collection  of pseudodisks,  let 
${{\cal U}^0}$ and ${{\cal V}^0}$ be any disjoint sub-decompositions of ${\cal U}$ and ${\cal V}$, respectively.
Let $U_i$ and $V_j$ be any two convex  pseudodisks from ${\cal U}$ and ${\cal V}$, respectively, and ${U^0_i}$ 
and ${V^0_j}$ be 
two corresponding convex objects  from  ${{\cal U}^0}$ and ${{\cal V}^0}$, respectively,   such that $\CF({U^0_i}, {{\cal U}^0}\cup {{\cal V}^0})\neq \emptyset$, 
$\CF({V^0_j}, {{\cal U}^0}\cup {{\cal V}^0})\neq \emptyset$ and 
$\interior({U^0_i})\cap 
\interior({V^0_j})\neq 
\emptyset$. Then we can find ${U^0_{ij}}\subseteq {U^0_i}$ 
and 
${V^0_{ji}}\subseteq {V^0_j}$ such that the following 
properties are 
satisfied.

\begin{itemize}

 \item[(i)] ${U^0_{ij}}$ and ${V^0_{ji}}$ are convex,  have nonempty disjoint interiors, and their intersection consists of a separating~line segment, which we denote by $E^0_{ij}$.

\item[(ii)] ${U^0_i}\setminus {U^0_{ij}}$ is completely 
contained in 
$V_j$.
\item[(iii)] ${V^0_j}\setminus {V^0_{ji}}$ is completely 
contained in $U_i$.
\end{itemize}

\end{lemma}
\begin{figure}
 
    \centering

  \includegraphics[page=1]{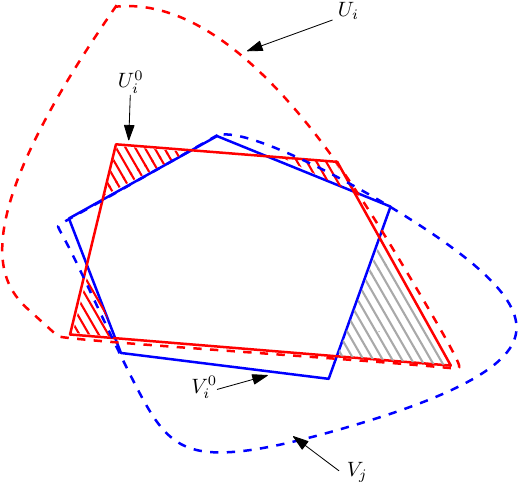}
   \caption{Petals: tiled regions are Petals of $U_i^0$; $\NCpetal$s are marked with red.}
   \label{figPetals}
\end{figure}

\begin{proof}

Given two convex   objects $U$ and $V$, define a \emph{petal} of $U$ with respect to $V$ to be a connected component of $U \setminus  V$. 
Since $U_i^0$ and $V_j^0$ need not be pseudodisks, there may be multiple petals of $U_i^0$ with respect to $V_j^0$. 
Let us assume that there are $k$ such petals, which we denote by $\petal_t(U_i^0)$, for $1\leq t\leq k$. 
Thus, $U_i^0 \setminus V_j^0 = \bigcup_{t=1}^k \petal_t(U_i^0)$. Similarly, we define $\petal(V_j^0)$ to be the set of petals of $V_j^0$ with respect to $U_i^0$ , and we let $k'$ denote their number.
Observe that each petal is bounded by two boundary arcs, one from $\partial U_i^0$ and the other from
$\partial V_j^0$ (see Figure~\ref{figPetals}). Also observe that consecutive petals are defined by consecutive intersection points between the boundaries of the two objects.

Since $V^0_j \subseteq V_j$, we have  ${U^0_i}\setminus V_j \subseteq {U^0_i}\setminus{V^0_j}$.
Define $\NCpetal({U^0_i})$ to be the subset of petals of $U_i^0$ (with respect to $V_j^0$) that are not entirely covered by $V_j$, that is, $\NCpetal({U^0_i})=\{\petal_t({U^0_i})\in \{{U^0_i}\setminus{V^0_j}\}|\petal_t({U^0_i})\cap \{{U^0_i}\setminus V_j\}\neq \emptyset\}$.
Similarly, we define  $\NCpetal({V^0_j})$. 
Because $\CF({U^0_i}, {{\cal U}^0}\cup {{\cal V}^0})\neq \emptyset$, $\NCpetal({U^0_i})$ contains at least one element, and the same holds for $\NCpetal({V^0_j})$ (see Figure~\ref{figPetals}).

Consider only the uncovered petals (that is, $\NCpetal(U_i^0) \cup \NCpetal(V^0_j)$). Let us label the
petals of $\NCpetal(U_i^0)$ with the letter ``u" and label the petals of $\NCpetal(V^0_j)$ with the letter
``v".  Let $R^0_{ij}={U^0_i}\cap {V^0_j}$. If you consider the cyclic order of these petals around $\partial R^0_{ij}$, 
the alternating pattern ``u\dots v\dots u\dots v” cannot occur in the cyclic sequence as shown in the following argument (see Figure~\ref{fig_petalLemma}).

\begin{itemize}
    \item[] 
Suppose to the contrary that the alternating pattern ``u\dots v\dots u\dots v” occurs in the cyclic sequence. Then there must exist points $u_1,u_2$ (from the first and third ``u" petals in the sequence) that lie
in $U_i^0 \setminus V^0_j$. Similarly, there exist points $v_1, v_2$ (from the second and fourth ``v" petals) that lie in $V^0_j \setminus U_i^0$. 
  Because of the alternation, the line segments
 $\overline{u_1u_2}$ and $\overline{v_1v_2}$ intersect in $R^0_{ij}$. 
However, the existence of these two line segments violates the hypothesis that $U_i$ and $V_j$ are pseudodisks.
\end{itemize}

Since the alternation pattern ``u\dots v\dots u\dots v” cannot arise in the cyclic sequence, it follows the cyclic order of uncovered petals around $\partial R^0_{ij}$ consists of a sequence of petals from $\NCpetal(U^0_i)$ followed by a sequence from $\NCpetal(V^0_j)$. As a result, we can find a line segment  $\overline{p_1p_2}$  lying 
in $\interior(R^0_{ij})$ whose two endpoints are on $\partial R^0_{ij}$ such that all the uncoverd petals of  $U^0_i$ (formally $\NCpetal(U^0_i)$) lie on one side of this line segment and the uncoverd petals of  $V^0_j$ (formally $\NCpetal(V^0_j)$) lie on the other side. 
In other words, extension of this line segment $\overline{p_1p_2}$ partitions the plane into two half-spaces ${\mathcal{H}}^0_i$ and ${\mathcal{H}}^0_j$ where ${\mathcal{H}}^0_i$ contains all the petals of $\NCpetal({U^0_i})$ and ${\mathcal{H}}^0_j$ contains all the petals of $\NCpetal({V^0_j})$. 
We define ${U^0_{ij}}={\mathcal{H}}^0_i\cap {U^0_i}$ and ${V^0_{ji}}={\mathcal{H}}^0_j\cap {V^0_j}$. 
The line segment  $\overline{{p}_1{p}_2}$  plays the role of the separating line segment $E^0_{ij}$. 
Claim~(i) follows because ${p}_1$ and  ${p}_2$ lie on the boundary of both $U_i^0$ and $V_j^0$. Claim~(ii) follows because $U_i^0 \setminus U _{ij}^0$ consists a portion of $R_{ij}^0$ (which clearly lies in $V_j$) together with a subset of petals of $U_i^0$ that are all covered by $V_j$ . Claim~(iii) is symmetrical. Hence ${U^0_{ij}}$, ${V^0_{ij}}$ satisfy the lemma. 
\begin{figure}
  \centering
   \includegraphics[width=\linewidth,scale=1.0,page=1]{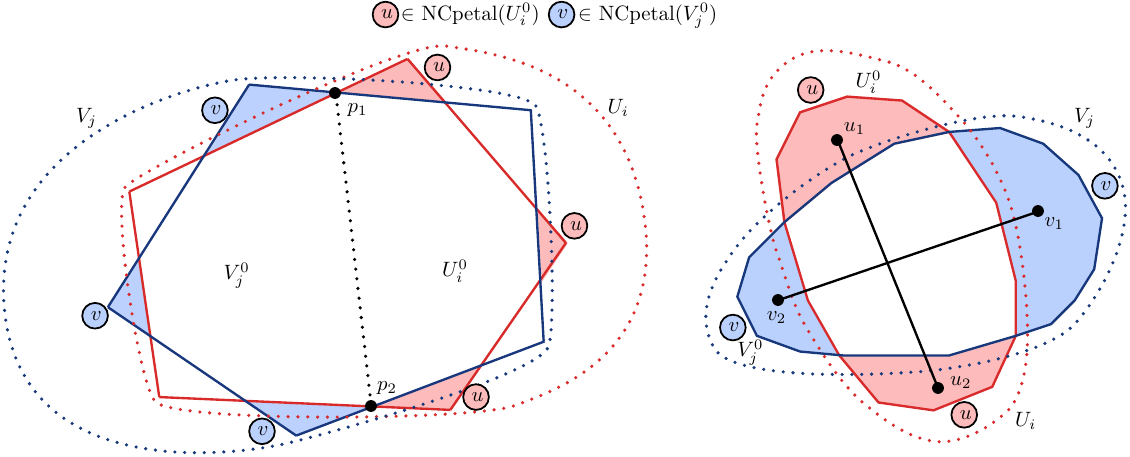}
   \caption{Illustration
   of Lemma~\ref{petalLemma}.}
   \label{fig_petalLemma}
\end{figure}
 \end{proof}

\section{Dominating-Set for Homothetic Convex Objects}\label{convPoly}
 
Let $C$ be a convex object in the plane.  We fix an arbitrary interior point of $C$ as 
the   {\it center}  $c(C)$. We are given a 
set $\mathscr{S}$ of $n$ homothetic (i.e., translated and uniformly scaled) copies of $C$, 
and our objective is to show that  the local-search algorithm given in Section~\ref{S-Alg} produces a PTAS for the minimum dominating-set for $\mathscr{S}$. 
Recall that ${\cal A}$  is  the set of objects returned by the  local-search algorithm, and  ${\cal O}$ is a minimum dominating-set. Without loss of generality, we assume that   both Claim~\ref{cD0.0} and \ref{cD0.1} are satisfied.  

In this section, we show mainly the existence of a planar graph 
satisfying the locality condition mentioned in Lemma~\ref{lem_loc_dom}.  Here is an overview of the proof.
First, we  find a disjoint sub-decomposition $\widetilde{\cal A}\cup \widetilde{\cal O}$ of ${\cal A}\cup {\cal O}$ (in Lemma~\ref{lDecompose2}).
Next, we  consider a nearest-site Voronoi diagram for the sites in $\widetilde{\cal A}\cup \widetilde{\cal O}$ with respect to a distance function.
Then we show (in Lemma~\ref{loc4}) that the dual of this  Voronoi diagram satisfies the locality condition mentioned in Lemma~\ref{lem_loc_dom}.  

\subsection{Decomposing into Interior Disjoint  Convex Sites}\label{SDecompose2}
Using  Lemmas~\ref{lDecompose1} and \ref{petalLemma} as tools, now we prove the following which is one of  the 
important observations of our work.

\begin{lemma}\label{lDecompose2}
Let  ${\cal A}$ be the output of the local-search algorithm for dominating-set on a set $\mathscr{S}$ of homothetic convex objects, and let   ${\cal O}$ be the optimum dominating-set. 
Then there exists
a disjoint sub-decomposition  $\widetilde{\cal A}\cup \widetilde{\cal O}$ which 
satisfies the following:
for any input  object $S\in \mathscr{S}$  
either
\begin{enumerate}[label=(\roman*)]
 \item  there exist  $\widetilde{A}\in \widetilde{\cal A}$ and 
$\widetilde{O}\in\widetilde{\cal O}$ such that  $S\cap \widetilde{A}\neq 
\emptyset$ and  $S\cap \widetilde{O}\neq \emptyset$, or
\item there exist 
 ${A}\in {\cal A}$ and 
${O}\in{\cal O}$ such that  $S\cap {A} \cap {O}\neq \emptyset$, and their traces
$\widetilde{A}$ and $\widetilde{O}$  share an edge  on their boundary.
\end{enumerate}
\end{lemma}
Remainder of this section is devoted to the proof of this lemma. As a continuation  from Section~\ref{anls}, we would like to remind the reader that  duplicate objects have been pruned from ${\cal A}$ and ${\cal O}$.
   
Let ${\cal A}=\{A_1,\ldots,A_{\ell}\}$  and ${\cal O}=\{O_1,\ldots,O_t\}$.
Our algorithm to obtain a disjoint
sub-decomposition  $\widetilde{\cal A}\cup \widetilde{\cal 
O}=\{\widetilde{A}_1,\ldots 
\widetilde{A}_{\ell}\}\cup\{\widetilde{O}_1,\ldots \widetilde{O}_{t}\}$ for 
${\cal A}\cup {\cal O}$ satisfying the lemma statement is  as follows.

{\bf Step 1: Obtaining  decompositions individually:} 
Note that the objects in ${\cal A}$ (resp., ${\cal 
O}$) are cover-free (follows from Claim~\ref{cD0.0}). So, we  apply  Lemma~\ref{lDecompose1} on the set ${\cal A}$ (resp., ${\cal 
O}$) of
objects, to compute 
the disjoint decomposition of ${\cal A}$ (resp., ${\cal {O}}$).
Let ${\cal A}^0=\{A^0_{1}, \ldots, 
A^0_{\ell}\}$ (resp., ${\cal O}^0=\{O_{1}^0, \ldots, 
O_{t}^0\}$ ) be the  disjoint decomposition of 
${\cal A}$ (resp., ${\cal {O}}$). Now, following claim is obvious.

\begin{claim}\label{cD0}
Any point $p\in \IR^2$ is contained in the interior of  at most two objects 
of  ${\cal 
A}^0\cup {\cal O}^0$.
\end{claim}

Lemma~\ref{lDecompose1} ensures that   $\CF({A_i},{\cal A})\subseteq A^0_i \neq \emptyset$ and
 $\CF({O_j},{\cal O})\subseteq O^0_j\neq \emptyset$ for all $i\in [\ell]$, $j\in[t]$.  
By Claim~\ref{cD0.1}, no object $A^0_i$ can be properly contained in any single object from ${\cal O}^0$, but it may be completely covered by the union of  two or more objects from ${\cal O}^0$. We can remedy this as follows.

Replace each object of ${\cal A}^0$ and ${\cal O}^0$ with an infinitesimally shrunken version of itself. By our assumption of general position, the resulting sets of shrunken objects still form dominating-sets. Furthermore, because the elements of ${\cal O}^0$ have pairwise disjoint interiors, no single object of ${\cal A}^0$ can be contained in the union of two or more of the shrunken objects in ${\cal O}^0$. Henceforth, ${\cal A}^0$ and ${\cal O}^0$ refer to the sets of shrunken objects. Thus we have the following.

\begin{claim}\label{cD0.3}
\begin{itemize}
\item[(i)] $\CF(A^0_i, {\cal A}^0\cup {\cal O}^0)\neq \emptyset$ for all $i\in[\ell]$,
\item[(ii)] $\CF(O^0_j, {\cal A}^0\cup {\cal O}^0)\neq \emptyset$ for all $j\in[t]$,

\item[(iii)] For each object $S\in {\cal S}$, there exist an  object  $A_i^0\in{\cal A}^0$ (resp., $O_j^0\in{\cal O}^0$)
such that $S\cap A_i^0 \neq \emptyset$ (resp., $S\cap O_j^0\neq \emptyset$).
\end{itemize}
\end{claim}

{\bf Step 2: Obtaining disjoint sub-decomposition:} Now, consider $A^0_{i}\in {\cal A}^0$ for all 
$i\in [\ell]$. 
Lemma~\ref{lDecompose1} ensures that $A^0_{i}$ does not have any interior overlap with $A^0_{k}$, for any $k \in [\ell]\setminus i$. 
Similarly, ${O^0_j}$ ($j\in [t]$) does not have any  interior overlap with $O_k^0$, for any $k \in [t]\setminus j$. 
But, $A^0_{i}$ may have interior overlap with one or more objects of ${\cal O}^0$.
Let $L(i)$ be the subset of indices $j\in [t]$ such that $A_i^0$ has an interior overlap with $O_j^0$. 
For any $j \in L(i)$, Claim~\ref{cD0.3} implies that both $\CF(A^0_i, {\cal A}^0\cup {\cal O}^0)\neq \emptyset$ and $\CF(O^0_j, {\cal A}^0\cup {\cal O}^0)\neq \emptyset$. 
By applying Lemma~\ref{petalLemma} to $A^0_{i}$ and 
${O^0_j}$, we obtain two interior-disjoint convex objects $A_{ij}^0 \subseteq A_i^0$ and $O_{ji}^0 \subseteq O_j^0$. Let $A^1_i=\bigcap_{j \in L(i)}A_{ij}^0$.
 Similarly, let $M(j)$ be the subset of indices $i \in [l]$ such that $O_j^0$ has an interior overlap with $A_i^0$.  Let $O^1_j=\bigcap_{i\in M(j)} O^0_{ji} $ which is a convex object and it contains  $\CF(O_j)$.
Let  ${\cal A}^1=\{A^1_{1}, \ldots, A^1_{\ell}\}$ and   ${\cal O}^1=\{{O^1_{1}}, \ldots, {O^1_{t}}\}$. 
Clearly, $A_i^1 \subseteq A_i^0$ and $O_j^1 \subseteq O_j^0$, and since separating line segments $E_{ij}^0$ have eliminated all overlaps between the two decompositions, it follows that ${\cal A}^1 \cup {\cal O}^1$ is a disjoint sub-decomposition of  ${\cal A}\cup {\cal O}$.
If we concentrate on  the arrangements of all $E^0_{ij}$ along the boundary of 
$\partial A^0_{i}$, then we  observe the following.

\begin{claim}\label{cD1}
Any two separating line segments $E^0_{ij}$ and $E^0_{ij'}$  do not intersect each 
other.
\end{claim}
\begin{proof}
  If  $E^0_{ij}$ and $E^0_{ij'}$ intersect each 
other then assertions (ii) and (iii) of Lemma~\ref{petalLemma} imply that the corresponding  objects 
${O^0_j}$ and  $O_{j'}^0$ also intersect, which is not possible because ${\cal O}^0$ is a disjoint  
decomposition. 
\end{proof}

The boundary $\partial A^1_{i}$ is actually obtained by replacing zero or more disjoint arcs of $\partial A^0_{i}$ with  separating line segments. Since each of these separating line segments are part of different disjoint objects in ${\cal O}^0$, here we would like to remark that the object $A^1_{i}$ is nonempty. For the similar reason,  each object $O^1_{j}\in {\cal O}^1$ is nonempty.
We denote the {\it partial boundary} $\Delta {A^0_{ij}}$ (resp., $\Delta {O^0_{ji}}$ ) by 
the portion 
of the boundary $\partial A^0_{i}$ (resp., 
$\partial {O^0_j}$) which is replaced by the edge   
$E^0_{ij}$ (see Figure~\ref{fig_empty}(b) where partial boundary is marked as 
dotted).

Note the following.

\begin{claim}\label{cD3}
Let $A^0_i$ and $O^0_j$ be any two  objects from ${\cal A}^0$ and ${\cal O}^0$, respectively, such that
$\interior({{A}^0_i})\cap 
\interior({{O}^0_j})\neq 
\emptyset$ and  $E^0_{ij}$  is not a part of $\partial A^1_i$. Then  following properties must be satisfied:

\begin{itemize}
\item there exists an object $O^0_{j'}$ in ${\cal O}^0$ such that $\interior({{A}^0_i})\cap 
\interior({{O}^0_{j'}})\neq 
\emptyset$,  $E^0_{i{j'}}$  is  a part of $\partial A^1_i$, and   ${A}^0_i \setminus {A}^0_{ij}$ is completely contained in ${O}_{j'}$.
\item $O^0_j$ does not intersect $A^1_i$.
\end{itemize}

\end{claim}
\begin{proof}
Claim~\ref{cD1} implies that that no two separating line segments intersect 
each other, {so the fact that $E^0_{ij}$ does not contribute to $\partial A^1_i$ implies that there is another object ${O^0_{j'}}$ such that the}
partial-boundary  $\Delta {A^0_{ij'}}$ contains the partial boundary 
$ \Delta{A^0_{ij}}$.  Thus, $A^0_{ij'}\subseteq A^0_{ij}$ which implies ${A}^0_i  \setminus  A^0_{ij} \subseteq {A}^0_i \setminus A^0_{ij'}$.  Since  $ {A}^0_i \setminus A^0_{ij'}$ is completely contained in ${O}_{j'}$ (by Lemma~\ref{petalLemma}), ${A}^0_i  \setminus  A^0_{ij}$ is also completely contained in ${O}_{j'}$.

Since ${O^0_{j}}$ and ${O^0_{j'}}$ are interior disjoint and the
partial-boundary  $\Delta {A^0_{ij'}}$ contains the partial boundary 
$ \Delta{A^0_{ij}}$, ${O^0_{j}}$ cannot intersect ${A^1_{i}}$.
Hence, the claim follows.
\end{proof}
{By a symmetrical argument,
we have the following.}
\begin{claim}\label{cD4}
Let $A^0_i$ and $O^0_j$ be any two  objects from ${\cal A}^0$ and ${\cal O}^0$, respectively, such that
$\interior({{A}^0_i})\cap 
\interior({{O}^0_j})\neq 
\emptyset$ and  $E^0_{ji}$  is not a part of $\partial O^1_j$. Then  following properties must be satisfied:

\begin{itemize}
\item there exists an object $A^0_{i'}$ in ${\cal A}^0$ such that $\interior({{O}^0_j})\cap 
\interior({{A}^0_{i'}})\neq 
\emptyset$,  $E^0_{j{i'}}$  is  a part of $\partial O^1_j$, and   ${O}^0_j \setminus {O}^0_{ji}$ is completely contained in ${A}_{i'}$.
\item $A^0_i$ does not intersect $O^1_j$.

\end{itemize}

\end{claim}

\begin{figure}
    \centering
    \begin{minipage}{.3\linewidth}
        \centering
       \includegraphics[width=\linewidth,scale=.08,page=1]{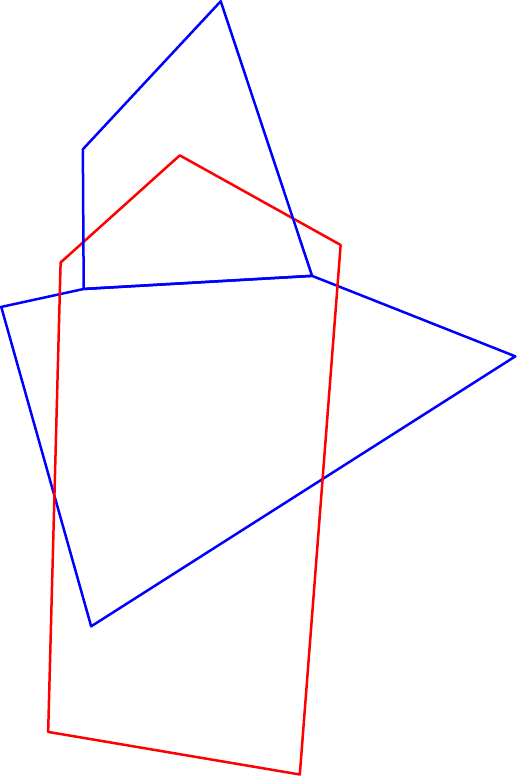}
	  \subcaption{After Step 1}
         \end{minipage}%
    \begin{minipage}{0.3\linewidth}
        \centering
            \includegraphics[width=\linewidth,scale=.08,page=2]{nfig2}
                 \subcaption{After Step 2 (Partial boundaries are shown dotted and the separating line segments are marked in green)}
    \end{minipage}
      \begin{minipage}{0.3\linewidth}
        \centering
                   \includegraphics[width=\linewidth,scale=.08,page=3]{nfig2} %
                      \subcaption{After Step 3}    
                              \end{minipage}
     \caption{Illustration of different steps: objects in $\cal A$ and $\cal O$ are marked with red and blue, respectively. }
      \label{fig_empty}
\end{figure}

Note that after this step, there might be some point $p\in {{A}^0_i}$ but $p\notin {{A}^1_i}$ and there does not exist any ${{O}^1_j}$ such that  $p \in {{O}^1_j}$  (see Figure~\ref{fig_empty}(a-b)). 
Hence, the objects of ${\cal A}^1 \cup {\cal O}^1$ fail to cover the same region as ${\cal A}^0 \cup {\cal O}^0$, as needed in the decomposition. 
To remedy this, we expand some of the objects in ${\cal A}^1$ and ${\cal O}^1$ in the next step.

{\bf Step 3: Expansion of objects in ${\cal A}^1$ and  ${\cal O}^1$:}

For each $(i, j) \in [\ell]\times [t]$, define $\chi(i, j) = 1$ if $E^0_{ij}$ is a part of $\partial A^1_i$ and $ E^0_{ji}$ is also a part of $\partial O^1_j$, and it is $0$ otherwise.  
Recalling $A^0_{ij}$ and $O^0_{ji}$ from Lemma~\ref{petalLemma}, for each $i\in[\ell]$, define $A^2_i=\bigcap\limits_{\{j| \chi(i, j) = 1\} } A^0_{ij} $, and for each $j\in [t]$,  define  $O^2_j=\bigcap\limits_{\{i| \chi(i, j) = 1\} }
O^0_{ji} $. 
Let  ${\cal A}^2=\{A^2_{1}, \ldots, 
A^2_{\ell}\}$ and  ${\cal O}^2=\{{O^2_{1}}, \ldots, 
{O^2_{t}}\}$. Note that ${\cal A}^2 \cup {\cal O}^2$ is a disjoint 
sub-decomposition of  ${\cal A}\cup {\cal O}$.
This construction 
along with Claims~\ref{cD3} and~\ref{cD4}  ensures the following.

\begin{claim}\label{cD9}
  \begin{itemize}
   \item For any point $p\in A^0_i\setminus A^2_i$, 
   there exists some ${O^2_j} \in {\cal O}^2$ such that  ${A^2_i}$ and ${O^2_j} $ share an edge on their boundary and $ p\in  O_j$.

\item For any point $p\in O^0_j\setminus O^2_j$, there exists  some ${A^2_i} \in {\cal A}^2$ such that  ${A^2_i}$ and  ${O^2 _j}$ share an edge on their boundary and $ p\in  A_i$.
\end{itemize}
\end{claim}

By renaming each set $A^2_i$ as $\widetilde{A}_i$ for $i \in [\ell]$ and each $O_j^2$ as $\widetilde{O}_j$ for $j \in [t]$, we obtain the final decomposition $\widetilde{\cal A}\cup \widetilde{\cal 
O}= {\cal A}^2\cup {\cal O}^2$. Finally, we claim the following which completes the proof of the lemma statement.

\begin{claim}\label{cD10}
 For any input  object $S\in \mathscr{S}$  
either (i) there exist  $\widetilde{A}\in \widetilde{\cal A}$ and 
$\widetilde{O}\in\widetilde{\cal O}$ such that  $S\cap \widetilde{A}\neq 
\emptyset$ and  $S\cap \widetilde{O}\neq \emptyset$, or (ii) there exist 
 ${A}\in {\cal A}$ and 
${O}\in{\cal O}$ such that  $S\cap {A} \cap {O}\neq \emptyset$, and 
$\widetilde{A}$ and $\widetilde{O}$  share an edge on their boundary.
\end{claim}

\begin{proof}
  Let $S$ be any input object in $\mathscr{S}$. From Claim~\ref{cD0.3} (iii), we know that  there exist 
$A^0_i\in {\cal A}^0$ and $O^0_j\in {\cal O}^0$ such 
that  $S\cap A^0_i\neq 
\emptyset$ and  $S\cap O^0_j\neq \emptyset$ for some $i\in[\ell]$ 
and $j\in [t]$. If after Step 3, $S\cap {A^2_i}\neq 
\emptyset$ and  $S\cap {O^2_j}\neq \emptyset$, then the claim 
follows. So without loss of generality assume that $S\cap {A^2_i}= 
\emptyset$. Consider any point $p\in S\cap A^0_i$. As $p\in  
A^0_i\setminus {A^2_i}$, there exist some ${O^2_j} \in  {\cal O}^2 $ 
such that ${A^2_i}$ and 
${O^2_j} $ share an edge on their boundary and $p\in O_{j}$ 
(follows from Claim~\ref{cD9}). Thus the claim follows.
\end{proof}

\subsection{Nearest-site Voronoi diagram} \label{nvd}

Recalling the definition of the convex distance function $\delta_{C}$ from Definition~\ref{def_1_cd}, we define the distance $\delta_{C}(p,P)$ from  a point $p$ to any 
 object $P$ (which need not be convex and homothetic to $C$) as follows.

\begin{definition}
 Let $p$ be a point and $P$ be an object in a plane. The distance  
$\delta_{C}(p,P)$ from   $p$ to  
 $P$ is defined as $\delta_{C}(p,P)=\min\limits_{q\in 
P}\delta_C(p,q)$.
\end{definition}

This distance function has the following properties.

\begin{property}\label{p5}
 \begin{itemize}
  \item[(i)] If $p$ is contained in the object $P$, then  
$\delta_{C}(p,P)=0$.
\item[(ii)] If $\delta_{C}(p,P)>0$, then $p$ is outside the object $P$, 
and a translated copy of $C$ centered at $p$ with 
 scaling factor $\delta_{C}(p,P)$  touches the object $P$.
 \end{itemize}

\end{property}

Now, we define a nearest-site Voronoi diagram~$\NVD_C$ for all 
the  objects in $\widetilde{\cal A}\cup \widetilde{\cal O}$ with respect 
to the distance 
function 
${\delta}_C$.
 
We define Voronoi cell of ${S}_i\in \widetilde{\cal A}\cup 
\widetilde{\cal O}$ as $\Cell({S}_i)=\{p\in \IR^2| {\delta}_C(p,{S}_i)\leq {\delta}_C(p,{S}_j) \textit{~for all } j\neq i\}$. The $\NVD_C$ is a partition on the plane 
 imposed by the collection of cells of all the objects in $\widetilde{\cal A}\cup 
\widetilde{\cal O}$.
A point $p$ is in $\Cell({S})$ for some object ${S}\in 
\widetilde{\cal A}\cup \widetilde{\cal O}$, 
implies that 
if we place a homothetic copy of $C$ centered at $p$ with a scaling factor 
$\delta_C(p,S)$, then $C$ touches  ${S}$ and the 
interior of $C$ is empty.
Now,  we have the following two lemmas.

 \begin{lemma}\label{lll1}
 The cell of every object ${S}\in \widetilde{\cal A}\cup 
\widetilde{\cal O}$ is nonempty. 
Moreover, 
  ${S} \subseteq  
\Cell({S})$. 
\end{lemma}
\begin{proof}

This follows from  Property~\ref{p5}(i) and the fact that $\widetilde{\cal A}\cup 
\widetilde{\cal O}$ is a set of 
interior disjoint  objects (from Lemma~\ref{lDecompose2}(a)).
\end{proof}

 \begin{lemma}\label{lll2}
 Each cell  $\Cell({S})$  is simply connected. 

\end{lemma}

\begin{proof}
For every $S\in \widetilde{\cal A}\cup 
\widetilde{\cal O}$, let us define the function $\pi_S \colon \mathbb{R}^2 \to S$, that maps any point to one of its closest points in $S$. (If $p \in S$, then $\pi_S(p) = p$.)

We first claim that for every point $p  \in \Cell(S)$, the line segment $\overline{p\pi_S(p)} \subseteq \Cell(S)$. To see this, suppose to the contrary that there exists a point $q \in\overline{p\pi_S(p)}$  such that $q\in \Cell(S')$ where $S('\neq S)\in \widetilde{\cal A}\cup 
\widetilde{\cal O}$. Then by basic
properties of convex distance functions (Property~\ref{prop:convex_dist}), we have
\[\delta_C(p, S')  \leq \delta_C(p, \pi_{S'}(q)) \leq \delta_C(p,q) + \delta_C(q, \pi_{S'}(q)) < \delta_C(p,q) + \delta_C(q, \pi_S(p)) = \delta_C(p, \pi_S(p)),\] contradicting the fact that $p\in \Cell(S)$.

To see that $\Cell(S)$ is connected, observe that any two points $p, p' \in \Cell(S)$ can be connected as follows. 
First, connect $p$ to $\pi_S(p)$ and $p'$ to $\pi_S(p')$. 
Then connect these two points through
$S$. 
By the above claim and Lemma~\ref{lll1}, all of these segments lies within $\Cell(S)$.

To complete the proof that $\Cell(S)$ is simply connected, we use the well known equivalent characterization~\cite{KleinW88} that for any simple closed (i.e., Jordan) curve $\Psi \subset \Cell(S)$, the interior of the region bounded by this curve lies entirely within $\Cell(S)$. 
Consider any $x$ in the interior of the region bounded by $\Psi$. 
Either $x \in S$ or (by extending the ray from $\pi_S(x)$ through $x$ until it hits $\Psi$) there exists $p \in \Cell(S)$ such that $x$ lies on the line segment $\overline{p\pi_S(x)}$. In the former case, $x\in \Cell(S)$, follows from Lemma~\ref{lll1}. Now, we are going to argue that $x\in \Cell(S)$ for the latter case as well. To see this, suppose to the contrary that $x\in \Cell(S')$ where $S('\neq S)\in \widetilde{\cal A}\cup 
\widetilde{\cal O}$.
Then by basic
properties of convex distance functions (Property~\ref{prop:convex_dist}), we have
\[\delta_C(p, S')  \leq \delta_C(p, \pi_{S'}(x)) \leq \delta_C(p,x) + \delta_C(x, \pi_{S'}(q)) < \delta_C(p,x) + \delta_C(x, \pi_S(p)) = \delta_C(p, \pi_S(p)),\] contradicting the fact that $p\in \Cell(S)$. Therefore $x \in \Cell(S)$, as desired. 
\end{proof}

\subsection{Locality Condition}
\label{loc-cond}
 
Let us consider the graph  ${\cal G}=({\cal V},{\cal E})$,  the \emph{dual} of the Voronoi diagram $\NVD_C$, whose  
  vertices $\mathcal{V}$ are the elements of $\mathcal{A} \cup \mathcal{O}$ and the edge set $\mathcal{E}$ consists of pairs $U, V \in \mathcal{V}$ whose Voronoi cells share an edge on their boundaries.  
 From 
Lemma~\ref{lll1} and Lemma~\ref{lll2},  we have the following.

\begin{lemma}
 The graph ${\cal G}=({\cal A}\cup {\cal O}, {\cal E})$  is a planar graph.
\end{lemma}
 
Now,  we prove that the graph ${\cal G}$ satisfies the property needed in the locality condition (Lemma~\ref{lem_loc_dom}).

\begin{lemma}\label{loc4}
  For any arbitrary input object $S  \in  \mathscr{S}$, if $S$ is dominated by at least one object of $\cal A$ and at least one object of $\cal O$, then there exists $A\in {\cal A}$ and $O\in {\cal O}$ both of which dominate $S$ and  $(A,O)\in {\cal E}$ of ${\cal G}$.
 \end{lemma}
 \begin{proof}
  Let $S$ be any object in $\mathscr{S}$.
According to Lemma~\ref{lDecompose2}, there exists a disjoint sub-decomposition
$\widetilde{\cal A} \cup \widetilde{\cal O}$ such that either:
\begin{itemize}
 \item[(i)]  there exist  $\widetilde{A}\in \widetilde{\cal A}$ and 
$\widetilde{O}\in\widetilde{\cal O}$ such that  $S\cap \widetilde{A}$ and  $S\cap \widetilde{O}$ are both nonempty, or 
 \item[(ii)]  there exist 
 ${A}\in {\cal A}$ and 
${O}\in{\cal O}$ such that  $S\cap {A} \cap {O}\neq \emptyset$, and their respective traces
$\widetilde{A}$ and $\widetilde{O}$ share an edge in common on their boundaries.
\end{itemize}

For case (ii), clearly both $A$ and $O$ dominates $S$. The fact that $\widetilde{A}$ and $\widetilde{O}$ share a common edge on their boundary implies (by Lemma~\ref{lll1}) that $\Cell(\widetilde{A})$ and $\Cell(\widetilde{O})$ also share a common edge on their boundaries. 
Therefore, $(A, O)$ is an edge of $\mathcal{G}$, as desired. 

For case (i), let $c = c(S)$ denote the center of $S$. 
Without loss of generality, we may assume that $A$ and $O$ have been chosen so that $\widetilde{A}$ and $\widetilde{O}$ are the closest objects to $c$ (with respect to $\delta_C$) in $\widetilde{\mathcal{A}}$ and $\widetilde{\mathcal{O}}$, respectively. 
We may assume that $\delta_C(c, \widetilde{A}) \leq \delta_C(c, \widetilde{O})$ (as the other case is symmetrical). 

Let $o \in \widetilde{O}$ denote the closest point to $c$ in $\widetilde{O}$. 
Clearly, $c$ and $o$ lie in different Voronoi cells, so this segment must intersect an edge of $\Cell(\widetilde{O})$ at some point $p$. Let $\Cell(\widetilde{R})$ denote the cell neighbouring the $\Cell(\widetilde{O})$ along this edge. Letting $r$ denote the closest point to $p$ in $\widetilde{R}$, we have $\delta_C(p, r) = \delta_C(p, \widetilde{R})= \delta_C(p, \widetilde{O})\leq \delta_C(p, o)$. By basic properties of convex distance function (see Property~\ref{prop:convex_dist}) we obtain 
\[ \delta_C(c,r)\leq \delta_C(c,p) + \delta_C(p,r) \leq \delta_C(c,p)+\delta_C(p,o) = \delta_C(c,o).\]
By general position, we may assume that $\delta_C(c, \widetilde{R}) < \delta_C(c, \widetilde{O})$.
Since $\widetilde{O}$ was chosen to be the closest object in $\widetilde{O}$ to $c$, it follows that $\widetilde{R} \in \widetilde{\mathcal{A}}$. 
Clearly, the associated objects $R$ and $O$ (which contain $\widetilde{R}$ and $\widetilde{O}$, respectively) both dominates $S$.
Therefore, there is an edge $(R, O)$ in $\mathcal{G}$, as desired. 
\end{proof}

\section{Dominating-Set for Homothets of a Centrally Symmetric Convex Object}\label{Appendix}
In this section, we give a simpler analysis of the local search algorithm for the dominating-set problem when 
the objects are homothets of a centrally symmetric  convex object.  Our analysis is a generalization of Gibson et al. \cite{GibsonP10} where we can avoid the sophisticated tool of disjoint decomposition. 

Let $C$ be a centrally symmetric convex object in the plane with the center $c(C)$.
Given a set  $\mathscr{S}$ of homothets  of  $C$, our objective is to show that  the local-search algorithm given in Section~\ref{S-Alg} is a PTAS for the minimum dominating-set for $\mathscr{S}$.  
  Recall  that ${\cal A}$  is  the set of objects returned by the  local-search algorithm, and  ${\cal O}$ is the minimum dominating-set. As a continuation from Section~\ref{S-Alg}, we assume that both Claim~\ref{cD0.0} and \ref{cD0.1} are satisfied.  

As in Section~\ref{nvd}, we define a nearest-site Voronoi diagram for all objects in ${\cal A}\cup {\cal O}$ with respect to a distance function $\delta_{C}^{*}$.
First, we are going to extend the convex distance function to provide meaningful (albeit negative) to the interior of each site. This would allow us to interpret the Voronoi diagram as a Voronoi diagram of additively weighted points, rather than a Voronoi diagram of (unweighted) regions. For each  object $S\in \mathscr{S}$, we define the \emph{weight} $w(S)$ to be $\alpha$, where $S=c(S)+\alpha C$.
  Now, we define the  distance $\delta_{C}^{*}(p,S)$ between a point $p \in \mathbb{R}^2$ and an object $S\in \mathscr{S}$ as follows: $\delta_{C}^{*}(p,S)=\delta_C(p,c({S}))-w(S)$. The  distance function $\delta_{C}^{*}(p,S)$ has the following properties:

\begin{property}
\label{p6}
\begin{enumerate}[label=(\roman*)]
\item The distance function $\delta_{C}^{*}(p,S)$ achieves its minimum value when $p=c(S)$.
\item If $p$ is contained in the object $S$, then
$\delta_{C}^{*}(p,S)\leq 0$.
\item If $\delta_{C}^{*}(p,S)>0$, then $p$ is outside the object $S$, 
and a translated copy of $C$ centered at $p$ with 
 scaling factor $\delta_{C}^{*}(p,S)$  touches the object $S$.
\end{enumerate}
\end{property}
Note that Property~\ref{p6}(iii) is crucial for our analysis and it follows due to the symmetric property of $\delta_C$. As a result, this approach cannot
be applied when objects are not centrally symmetric.

 We will show that each object in ${\cal A}\cup {\cal O}$ has a nonempty cell in this Voronoi diagram and each cell is simply connected. As a result the graph  ${\cal G}=({\cal V},{\cal E})$ which  is the dual of this Voronoi diagram is planar. Finally, we will show that this graph satisfies the locality condition mentioned in Lemma~\ref{lem_loc_dom}. This completes the proof.

\begin{lemma}\label{Al1}
 The cell of every object ${S}\in {\cal A}\cup 
{\cal O}$ is nonempty. 
Moreover, 
 the center $c({S}) \subseteq  
\Cell({S})$. 
\end{lemma}
\begin{proof}
For the sake of contradiction, assume for some object $S\in{\cal A}\cup {\cal O}$, $c(S)\notin \Cell(S)$ and $c(S)\in \Cell(S')$ where $S'(\neq S)\in {\cal A}\cup {\cal O}$. So, $\delta_{C}^{*}(c(S),S)\geq \delta_{C}^{*}(c(S),S')$. Since $\delta_{C}^{*}(c(S),S)=-w(S)$, we have $-w(S)\geq \delta_{C}(c(S),c(S'))-w(S')$. This implies $w(S')\geq\delta_{C}(c(S),c(S')) +w(S)$ which means that the object $S$ is contained in the object $S'$. This contradicts Claim~\ref{cD0.0} and~\ref{cD0.1}. 
\end{proof}

 \begin{lemma}\label{Al2}
 Each cell  $\Cell({S})$  is simply connected.
 \end{lemma}

\begin{proof}
We first claim that for every point $p  \in \Cell(S)$, the line segment $\overline{pc(S)} \subseteq \Cell(S)$. To see this, suppose to the contrary that there exists a point $q \in\overline{pc(S)}$  such that $q\in \Cell(S')$ where $S'(\neq S)\in {\cal A}\cup {\cal O}$. Then by basic
properties of convex distance functions (Property~\ref{prop:convex_dist}), we have
\[\delta_C^{*}(p, S')= \delta_C(p,c(S'))-w(S') \leq \delta_C(p,q)+ \delta_C(q,c(S'))-w(S')
\leq \delta_C(p,q) + \delta_C^{*}(q, S')\] \[ < \delta_C(p,q) + \delta_C^{*}(q, S) =\delta_C(p,q) + \delta_C(q,c(S))-w(S)= \delta_C(p,c(S))-w(S)= \delta_C^{*}(p, S),\] contradicting the fact that $p\in \Cell(S)$.

To see that $\Cell(S)$ is connected, observe that any two points $p, p' \in \Cell(S)$ can be connected via  $c(S)$ as follows. 
First, connect $p$ to $c(S)$  and then connect $p'$ to $c(S)$. 
 By the above claim and Lemma~\ref{Al1}, all of these segments lies within $\Cell(S)$.

To complete the proof that $\Cell(S)$ is simply connected, we use the well known equivalent characterization~\cite{KleinW88} that for any simple closed (i.e., Jordan) curve $\Psi \subset \Cell(S)$, the interior of the region bounded by this curve lies entirely within $\Cell(S)$. 
Consider any $x$ in the interior of the region bounded by $\Psi$. 
Either $x= c(S)$ or (by extending the ray from $c(S)$ through $x$ until it hits $\Psi$) there exists $p \in \Cell(S)$ such that $x$ lies on the line segment $\overline{pc(S)}$. In the former case, $x\in \Cell(S)$, follows from Lemma~\ref{Al1}. For the latter case, by the above claim (that $\overline{pc(S)} \subseteq \Cell(S)$), we have  $x\in \Cell(S)$. This completes the proof.
\end{proof}

\begin{lemma}\label{Aloc4}
  For any arbitrary input object $S  \in  \mathscr{S}$, there is an edge 
between 
$(A,O)\in {\cal G}$ such that $A\in {\cal A}$ and $O\in {\cal O}$, and both 
$A$ and $O$ dominates $S$. 
 \end{lemma}
 \begin{proof}
The proof is similar to the Case (i) of Lemma~\ref{loc4}.
\end{proof}


 \section{Geometric Set-Cover for Convex Pseudodisks}\label{setCover}
 
 Given a set  $\mathscr{S}$ of $n$ convex  pseudodisks and a set $\cal P$ of points in $\IR^2$, the objective is to cover all the points in $\cal P$ using subset of $\mathscr{S}$ of minimum cardinality.
Here, we analyze that the local search algorithm, as given in Section~\ref{S-Alg}, would give a polynomial time approximation scheme. 
The analysis is similar to the previous problem. 
Recall from Section~\ref{anls} that   ${\cal O}$ is an optimal covering set for $\cal P$ and ${\cal A}$ 
is  the covering set returned by our local search algorithm satisfying both  Claim~\ref{cD0.0} and \ref{cD0.1}. 
Here, we need to show that the locality condition mentioned in 
Lemma~\ref{lem_loc_cov}  is satisfied.

If we restrict the proof  of Lemma~\ref{lDecompose2} up to Claim~\ref{cD9}, then, it is straightforward to obtain the following.

\begin{lemma}\label{set_cover}
Let ${\cal A}$ be the output of the local-search algorithm for set-cover on a set $\mathscr{S}$ of 
convex  pseudodisks and a set $\cal P$ 
of points in 
$\IR^2$,  and let   ${\cal O}$ be the optimum. 
Then  there exists a 
disjoint sub-decomposition  $\widetilde{\cal A}\cup \widetilde{\cal O}$ 
which 
satisfies the following: for any input point $p\in {\cal P}$  
 there exist 
 ${A}\in {\cal A}$ and 
${O}\in{\cal O}$ such that  $p\in {A}$ and $p \in {O}$, and their traces
$\widetilde{A}$ and $\widetilde{O}$  share an edge on their boundary.
\end{lemma}



\begin{proof}

Let ${\cal A}=\{A_1,\ldots,A_{\ell}\}$  and ${\cal O}=\{O_1,\ldots,O_t\}$. Our algorithm to obtain a disjoint
sub-decomposition  $\widetilde{\cal A}\cup \widetilde{\cal 
O}=\{\widetilde{A}_1,\ldots 
\widetilde{A}_{\ell}\}\cup\{\widetilde{O}_1,\ldots \widetilde{O}_{t}\}$ for 
${\cal A}\cup {\cal O}$ satisfying the lemma statement is  exactly same as the three steps mentioned in Section~\ref{SDecompose2} for Lemma~\ref{lDecompose2}. 
The main difference is in the statement of Claim~\ref{cD4}. For set-cover problem, we have the following 

\begin{claim}\label{sD0.3}
\begin{itemize}
\item[(i)] $\CF(A^0_i, {\cal A}^0\cup {\cal O}^0)\neq \emptyset$ for all $i\in[\ell]$,
\item[(ii)] $\CF(O^0_j, {\cal A}^0\cup {\cal O}^0)\neq \emptyset$ for all $j\in[t]$,

\item[(iii)] Each point $p\in {\cal P}$ is covered by exactly one object from  ${\cal A}^0$ (resp., ${\cal O}^0$).
\end{itemize}
  
\end{claim}

Finally, instead of Claim~\ref{cD10}, we claim the following statement.

\begin{claim}\label{sD10_c}
 For any input point $p \in {\cal P}$,  there exist 
 ${A}\in {\cal A}$ and 
${O}\in{\cal O}$ such that  $p\in {A}$ and $p \in {O}$, and 
$\widetilde{A}$ and $\widetilde{O}$  share an edge on their boundary.
\end{claim}

\begin{proof}
  Let $p$ be any input point in ${\cal P}$. By Claim~\ref{sD0.3} (iii), there exist 
$A^0_i\in {\cal A}^0$ and $O^0_j\in {\cal O}^0$ such 
that  $p \in A^0_i$ and  $p \in O^0_j$ 
for some $i\in[\ell]$ 
and $j\in [t]$. After Step 3, since  ${\cal 
A}^2\cup {\cal O}^2$ is a disjoint decomposition of ${\cal A}\cup {\cal O}$,  $p$ cannot be both in 
${A^2_i}$ and ${O^2_j}$.  Therefore,    either of the 
following happens: $p \notin 
{A^2_i}$, or $p \notin {O^2_j}$.
In both  cases, the claim follows  from Claim~\ref{cD9}.
\end{proof}   

Thus the lemma follows.
  \end{proof}


Now,  consider  a graph ${\cal G}=({\cal V},{\cal E})$, where  each vertex 
$V\in {\cal V}$ corresponds to an object in $\widetilde{\cal A}\cup 
\widetilde{\cal O}$,  and we create an edge in between two vertices whenever the corresponding objects in $\widetilde{\cal A}\cup \widetilde{\cal O}$ share an edge in their boundary. 
Since, the objects of $\widetilde{\cal A}\cup \widetilde{\cal O}$ are convex and have disjoint interiors, this graph is a  planar graph.  
From Lemma~\ref{set_cover}, it follows that  the graph ${\cal G}$ satisfies the {\it locality condition} mentioned in  Lemma~\ref{lem_loc_cov}.
This completes the proof of Theorem~\ref{Thm:SetCover}.

\section{Concluding Remarks}
In this paper, we have shown that the well-known local search algorithm gives a PTAS for  finding the minimum cardinality dominating-set and geometric set-cover when the objects are homothetic convex objects, and  convex pseudodisks, respectively. 
As a consequence, we obtain easy to implement approximation guaranteed algorithms for a broad class of objects which encompasses arbitrary squares, $k$-regular polygons,  translates of  convex polygons. 
A QPTAS is known for the weighted set-cover problem where objects are pseudodisks~\cite{Ray}. 
But,  no QPTAS is known for the weighted dominating-set problem when objects are homothetic convex objects.  
Note that the separator-based arguments for finding PTAS has a limitation for handling the weighted version of the problems. 
Thus, finding a polynomial time approximation scheme for the weighted version of both minimum dominating-set and minimum geometric set-cover problems for  homothetic convex objects, pseudodisks remain open in this context. 
Specially, for the weighted version of the problem, it would be interesting to analyze the approximation guarantees of local search algorithm.

\section*{Acknowledgement}
The authors would like to acknowledge the reviewer for constructive feedback that improves the quality of the paper significantly.
We thank Mr A. B. Roy  for participating in the early discussions of this work.


\begin{thebibliography}{10}

\bibitem{AdamaszekW14}
Anna Adamaszek and Andreas Wiese.
\newblock A {QPTAS} for maximum weight independent set of polygons with
  polylogarithmically many vertices.
\newblock In {\em Proceedings of the Twenty-Fifth Annual {ACM-SIAM} Symposium
  on Discrete Algorithms, {SODA} 2014, Portland, Oregon, USA, January 5-7,
  2014}, pages 645--656, 2014.

\bibitem{BansalP14}
Nikhil Bansal and Kirk Pruhs.
\newblock The geometry of scheduling.
\newblock {\em {SIAM} J. Comput.}, 43(5):1684--1698, 2014.

\bibitem{ChanG}
Timothy~M. Chan and Elyot Grant.
\newblock Exact algorithms and {APX}-hardness results for geometric packing and
  covering problems.
\newblock {\em Comput. Geom.}, 47(2):112--124, 2014.

\bibitem{ChanH09}
Timothy~M. Chan and Sariel Har{-}Peled.
\newblock Approximation algorithms for maximum independent set of pseudo-disks.
\newblock In {\em Proceedings of the 25th {ACM} Symposium on Computational
  Geometry, Aarhus, Denmark, June 8-10, 2009}, pages 333--340, 2009.

\bibitem{Chekuri}
Chandra Chekuri, Kenneth~L. Clarkson, and Sariel Har{-}Peled.
\newblock On the set multicover problem in geometric settings.
\newblock {\em {ACM} Transactions on Algorithms}, 9(1):9, 2012.

\bibitem{ChewD85}
L.~Paul Chew and Robert L. (Scot)~Drysdale III.
\newblock Voronoi diagrams based on convex distance functions.
\newblock In {\em Proceedings of the First Annual Symposium on Computational
  Geometry, Baltimore, Maryland, USA, June 5-7, 1985}, pages 235--244, 1985.

\bibitem{ClarksonV07}
Kenneth~L. Clarkson and Kasturi~R. Varadarajan.
\newblock Improved approximation algorithms for geometric set cover.
\newblock {\em Discrete {\&} Computational Geometry}, 37(1):43--58, 2007.

\bibitem{Cohen-AddadM15}
Vincent Cohen{-}Addad and Claire Mathieu.
\newblock Effectiveness of local search for geometric optimization.
\newblock In {\em 31st International Symposium on Computational Geometry, SoCG
  2015, June 22-25, 2015, Eindhoven, The Netherlands}, pages 329--343, 2015.

\bibitem{DinurS14}
Irit Dinur and David Steurer.
\newblock Analytical approach to parallel repetition.
\newblock In {\em Symposium on Theory of Computing, {STOC} 2014, New York, NY,
  USA, May 31 - June 03, 2014}, pages 624--633, 2014.

\bibitem{Edelsbrunner95}
Herbert Edelsbrunner.
\newblock The union of balls and its dual shape.
\newblock {\em Discrete {\&} Computational Geometry}, 13:415--440, 1995.

\bibitem{ErlebachL08}
Thomas Erlebach and Erik~Jan van Leeuwen.
\newblock Domination in geometric intersection graphs.
\newblock In {\em {LATIN} 2008: Theoretical Informatics, 8th Latin American
  Symposium, B{\'{u}}zios, Brazil, April 7-11, 2008, Proceedings}, pages
  747--758, 2008.

\bibitem{ErlebachL10}
Thomas Erlebach and Erik~Jan van Leeuwen.
\newblock {PTAS} for weighted set cover on unit squares.
\newblock In {\em Approximation, Randomization, and Combinatorial Optimization.
  Algorithms and Techniques, 13th International Workshop, {APPROX} 2010, and
  14th International Workshop, {RANDOM} 2010, Barcelona, Spain, September 1-3,
  2010. Proceedings}, pages 166--177, 2010.

\bibitem{Feige98}
Uriel Feige.
\newblock A threshold of ln \emph{n} for approximating set cover.
\newblock {\em J. {ACM}}, 45(4):634--652, 1998.

\bibitem{Frederickson87}
Greg~N. Frederickson.
\newblock Fast algorithms for shortest paths in planar graphs, with
  applications.
\newblock {\em {SIAM} J. Comput.}, 16(6):1004--1022, 1987.

\bibitem{GareyJ79}
M.~R. Garey and David~S. Johnson.
\newblock {\em Computers and Intractability: {A} Guide to the Theory of
  NP-Completeness}.
\newblock W. H. Freeman, 1979.

\bibitem{GibsonP10}
Matt Gibson and Imran~A. Pirwani.
\newblock Algorithms for dominating set in disk graphs: Breaking the
  log\emph{n} barrier - (extended abstract).
\newblock In {\em Algorithms - {ESA} 2010, 18th Annual European Symposium,
  Liverpool, UK, September 6-8, 2010. Proceedings, Part {I}}, pages 243--254,
  2010.

\bibitem{HarPel12}
Sariel Har{-}Peled and Mira Lee.
\newblock Weighted geometric set cover problems revisited.
\newblock {\em JoCG}, 3(1):65--85, 2012.

\bibitem{Har-PeledQ15}
Sariel Har{-}Peled and Kent Quanrud.
\newblock Approximation algorithms for polynomial-expansion and low-density
  graphs.
\newblock In {\em Algorithms - {ESA} 2015 - 23rd Annual European Symposium,
  Patras, Greece, September 14-16, 2015, Proceedings}, pages 717--728, 2015.

\bibitem{HochbaumM87}
Dorit~S. Hochbaum and Wolfgang Maass.
\newblock Fast approximation algorithms for a nonconvex covering problem.
\newblock {\em J. Algorithms}, 8(3):305--323, 1987.

\bibitem{Viggo-Kann-Thesis}
Viggo Kann.
\newblock {\em On the Approximability of {NP}-complete Optimization Problems}.
\newblock PhD thesis, Royal Institute of Technology, 1992.

\bibitem{Karp}
Richard~M Karp.
\newblock Reducibility among combinatorial problems.
\newblock In {\em Complexity of Computer Computations}, pages 85--103, 1972.

\bibitem{KelleyN}
J.L. Kelley and I.~Namioka.
\newblock {\em Linear Topological Spaces}.
\newblock Springer-Verlag, 1976.

\bibitem{KleinW88}
Rolf Klein and Derick Wood.
\newblock Voronoi diagrams based on general metrics in the plane.
\newblock In {\em {STACS} 88, 5th Annual Symposium on Theoretical Aspects of
  Computer Science, Bordeaux, France, February 11-13, 1988, Proceedings}, pages
  281--291, 1988.

\bibitem{LenzenW10}
Christoph Lenzen and Roger Wattenhofer.
\newblock Minimum dominating set approximation in graphs of bounded arboricity.
\newblock In {\em Distributed Computing, 24th International Symposium, {DISC}
  2010, Cambridge, MA, USA, September 13-15, 2010. Proceedings}, pages
  510--524, 2010.

\bibitem{LiJ}
Jian Li and Yifei Jin.
\newblock A {PTAS} for the weighted unit disk cover problem.
\newblock In {\em Automata, Languages, and Programming - 42nd International
  Colloquium, {ICALP} 2015, Kyoto, Japan, July 6-10, 2015, Proceedings, Part
  {I}}, pages 898--909, 2015.

\bibitem{Marx06}
D{\'{a}}niel Marx.
\newblock Parameterized complexity of independence and domination on geometric
  graphs.
\newblock In {\em Parameterized and Exact Computation, Second International
  Workshop, {IWPEC} 2006, Z{\"{u}}rich, Switzerland, September 13-15, 2006,
  Proceedings}, pages 154--165, 2006.

\bibitem{Marx08}
D{\'{a}}niel Marx.
\newblock Parameterized complexity and approximation algorithms.
\newblock {\em Comput. J.}, 51(1):60--78, 2008.

\bibitem{Ray}
Nabil~H. Mustafa, Rajiv Raman, and Saurabh Ray.
\newblock Quasi-polynomial time approximation scheme for weighted geometric set
  cover on pseudodisks and halfspaces.
\newblock {\em {SIAM} J. Comput.}, 44(1):1650--1669, 2015.

\bibitem{MustafaR10}
Nabil~H. Mustafa and Saurabh Ray.
\newblock Improved results on geometric hitting set problems.
\newblock {\em Discrete {\&} Computational Geometry}, 44(4):883--895, 2010.

\bibitem{RazS97}
Ran Raz and Shmuel Safra.
\newblock A sub-constant error-probability low-degree test, and a sub-constant
  error-probability {PCP} characterization of {NP}.
\newblock In {\em Proceedings of the Twenty-Ninth Annual {ACM} Symposium on the
  Theory of Computing, STOC 1997, El Paso, Texas, USA, May 4-6, 1997}, pages
  475--484, 1997.

\bibitem{Leeuwen-Thesis}
Erik~Jan van Leeuwen.
\newblock {\em Optimization and Approximation on Systems of Geometric Objects}.
\newblock PhD thesis, University of Amsterdam, 2009.

\bibitem{Varada}
Kasturi~R. Varadarajan.
\newblock Weighted geometric set cover via quasi-uniform sampling.
\newblock In {\em Proceedings of the 42nd {ACM} Symposium on Theory of
  Computing, {STOC} 2010, Cambridge, Massachusetts, USA, 5-8 June 2010}, pages
  641--648, 2010.

\end{thebibliography}
\end{document}